\documentclass[draftclsnofoot,onecolumn,12pt]{IEEEtran} %

\usepackage{cite}
\usepackage{amsmath,amssymb,amsfonts}
\usepackage{graphicx}
\usepackage{textcomp}

\usepackage{bm}

\usepackage{amsthm}
\usepackage{cleveref} 
\usepackage{float}
\usepackage{tikz}
\usetikzlibrary{arrows.meta, calc, positioning, shapes.geometric, shadows}
\usepackage{algorithm}
\usepackage{algpseudocode} 
\usepackage{tabularx}
\usepackage{array}
\newcolumntype{Y}{>{\raggedright\arraybackslash}X}
\newcolumntype{C}[1]{>{\centering\arraybackslash}p{#1}} 
\usepackage{url}
\newcommand{\Pin}{P_{\mathrm{in}}}     % input tone power (non-ideal sections)
\newcommand{\Ts}{T_{\mathrm{s}}}       % symbol duration (symbol-wise model)

\newtheorem{lemma}{Lemma}
\newtheorem{theorem}{Theorem}
\newtheorem{corollary}{Corollary}
\newtheorem{remark}{Remark}
\newtheorem{definition}{Definition}

\crefname{algorithm}{Algorithm}{Algorithms}
\Crefname{algorithm}{Algorithm}{Algorithms}

\algrenewcommand\algorithmicrequire{\textbf{Input:}}
\algrenewcommand\algorithmicensure{\textbf{Output:}}

\begin{document}

\title{A Narrowband Fully-Analog Multi-Antenna Transmitter}

\author{Nikola~Zlatanov%
\thanks{N.~Zlatanov is with Innopolis University, Innopolis, 420500, Russia (e-mail: n.zlatanov@innopolis.ru).}%
}

\maketitle

\begin{abstract}
 This paper proposes a narrowband fully-analog $N$-antenna transmitter that emulates the functionality of a narrowband fully-digital $N$-antenna transmitter. Specifically, in symbol interval $m$, the proposed fully-analog transmitter synthesizes an arbitrary complex excitation vector $\bm x[m]\in\mathbb{C}^N$ with prescribed total power $\|\bm x[m]\|_2^2=P$ from a single coherent RF tone, using only tunable phase-control elements embedded in a passive interferometric programmable network. The programmable network is excited through one input port while the remaining $N\!-\!1$ input ports are impedance matched. In the ideal lossless case, the network transfer is unitary and therefore redistributes RF power among antenna ports without dissipative amplitude control.

The synthesis task is posed as a unitary state-preparation problem: program a unitary family so that
$\bm V(\bm\varphi[m])\bm e_1=\bm c[m]$, where $\bm c[m]=\bm x[m]/\sqrt{P}$ and $\|\bm c[m]\|_2=1$. We provide a constructive, parameter-minimal split-then-phase realization and a closed-form programming rule: a balanced binary magnitude-splitting tree allocates the desired per-antenna magnitudes $|c_n|$ using $N\!-\!1$ tunable split ratios, and a per-antenna output phase bank assigns the target phases using $N$ tunable phase shifts. For $N$ a power of two ($N=2^L$), the magnitude tree uses $N-1$ interferometric $2\times2$ fanout cells arranged in $L=\log_2 N$ layers. The resulting architecture uses exactly $2N-1$ real tunable degrees of freedom and admits a deterministic $O(N)$ programming procedure with no iterative optimization, enabling symbol-by-symbol updates when the chosen phase-control technology supports the required tuning speed.

Using representative  commercial off-the-shelf (COTS) components, we model the compute-excluded RF-front-end DC power of the proposed fully-analog transmitter and compare it against an equivalent  COTS fully-digital array. For $N\le 16$, the comparison indicates significant RF-front-end power savings for the fully-analog architecture under a common delivered antenna-port power normalization.

The results are established under an ideal, lossless, and perfectly matched network model and are intended as a proof-of-concept for a narrowband fully-analog transmitter. Detailed treatment of hardware non-idealities such as finite phase resolution, tuning-dependent loss/imbalance, parasitics, calibration procedures, and over-the-air validation is left to future experimental work.
\end{abstract}


 % ============================================================
% Section 1: Introduction
% ============================================================
\section{Introduction}
\label{sec:introduction}

Large antenna arrays enable high beamforming gain, spatial selectivity, and improved link budgets,
which are central to modern wireless systems (e.g., mmWave) and dense deployments. In current
practice, high-performance multi-antenna transmitters are typically realized via \emph{multiple RF
chains} (fully-digital transmitter) or \emph{hybrid} analog/digital front ends. Both approaches
scale hardware complexity with the number of antennas through mixers, local-oscillator distribution,
digital-to-analog converters (DACs), filtering, and per-branch calibration, which become
increasingly costly and power-hungry as $N$ grows; see, e.g., \cite{mendez2016hybrid}. These
realities motivate \emph{fully-analog} transmitter architectures in which a large subset of this
hardware can be eliminated.

A core obstacle for fully-analog architectures is the synthesis of \emph{arbitrary} complex
per-antenna excitations. Phase-shifter-only networks naturally provide phase control at
each branch, but \emph{independent} amplitude control is not available. Achieving both amplitude and
phase control typically requires either (i) additional active elements such as attenuators or
variable-gain components, and/or (ii) magnitude shaping via constructive and destructive summation of
multiple tones with adjusted phases, which can introduce loss and sensitivity to hardware
non-idealities \cite{mailloux2005phasedarray}.

In this work we propose a narrowband fully-analog $N$-antenna transmitter that, in symbol interval $m$, synthesizes an
arbitrary complex excitation vector $\bm x[m]\in\mathbb{C}^N$ with prescribed total power
$\|\bm x[m]\|_2^2=P$ from a single coherent RF tone using a passive interferometric programmable
network. The network is excited through one input port while the remaining inputs are impedance
matched. The key design principle is to enforce an (ideal) \emph{unitary} transfer inside the
programmable network so that RF power is \emph{redistributed among ports} rather than dissipated
internally. In the ideal lossless model, this yields exact synthesis of any target vector satisfying
$\|\bm x[m]\|_2^2=P$ without programmable attenuators or active gain stages.

We target symbol-wise operation: during each symbol interval the desired complex vector $\bm x[m]$ (as
determined by the baseband precoder and the transmitted data symbols) is realized by programming the
tunable controls of the passive network. 
Accordingly, aside from the  total-power constraint, there is no restriction on the
sequence $\{\bm x[m]\}$: the network settings   can be updated independently each
symbol interval to realize any desired $\bm x[m]$.
Consequently, the achievable symbol rate is limited by the
update and settling time of the chosen phase-control technology and its control interface, rather
than by iterative optimization latency. The programming rule developed here is explicitly
constructive, closed form, and requires only $O(N)$ scalar operations to compute the $2N-1$ tunable
degrees of freedom needed to realize an arbitrary unit-norm excitation direction.

A central ingredient is a parameter-minimal \emph{split-then-phase} state-preparation realization.
Writing $\bm c[m]=\bm x[m]/\sqrt{P}$ and $\bm c[m]=[|c_1[m]|e^{j\angle c_1[m]},\dots,|c_N[m]|e^{j\angle c_N[m]}]^T$ with
$\|\bm c[m]\|_2=1$, the proposed programmable network is organized as:
(i) a balanced binary magnitude-splitting tree that allocates the desired per-antenna magnitudes
$|c_n[m]|$ using only tunable split ratios at internal nodes, followed by
(ii) a per-antenna output phase bank that assigns the target phases $\angle c_n[m]$.
The magnitude tree is realized by a binary tree of interferometric $2\times2$ fanout cells (e.g.,
MZI/hybrid-based) operating in a single-input ``fanout'' mode where one input is matched and
unexcited; this operating mode admits a deterministic, closed-form split programming rule based on
subtree norms. Because each $2\times2$ splitter introduces a known fixed branch phase convention
(e.g., a quadrature $+\pi/2$ offset on one branch under an MZI model), the magnitude tree produces
deterministic per-output phase offsets that depend only on the root-to-leaf path. These offsets are
then compensated by the output phase bank. The resulting architecture uses exactly $N-1$ split
controls and $N$ output phase controls, totaling $2N-1$ real tunable degrees of freedom, which
matches the real dimension of the complex unit sphere and is therefore parameter-minimal for
universal single-vector synthesis.

For clarity of exposition, the main development focuses on $N$ as a power of two ($N=2^L$), which
enables a balanced binary-tree implementation with depth $L=\log_2 N$; the general non-power-of-two case is
handled later via zero-padding and optional tree pruning.

An important implication is that the proposed architecture is \emph{spatially fully-digital-equivalent} in the
narrowband sense: since it can synthesize an arbitrary complex transmit vector $\bm x[m]\in\mathbb{C}^N$
(subject to $\|\bm x[m]\|_2^2=P$) at symbol cadence, it is not restricted to rank-one analog beamforming.
Instead, it can execute the same narrowband linear-precoding operations as a fully-digital multi-antenna transmitter,
including spatial multiplexing and multi-user MIMO, by realizing $\bm x[m]=\bm W\,\bm d[m]$ for any desired
precoder $\bm W$ and data vector $\bm d[m]$, in each symbol interval $m$. The resulting system-level trade-off is therefore primarily
\emph{bandwidth versus RF-chain energy}: the synthesis is narrowband (single tone or single OFDM subcarrier), but
the elimination of $N$ RF chains yields substantially lower RF-front-end DC power, which can improve energy
efficiency and partially offset the throughput penalty of narrowband operation when MIMO gains are exploited.

\subsection{Related work and the gap}
\label{subsec:related_work}

Classical analog beamforming networks (BFNs) include corporate-feed dividers, Butler (DFT) matrices,
and lens-based networks, which provide low-loss distribution and beam steering with varying degrees
of flexibility \cite{mailloux2005phasedarray}. Hybrid beamforming reduces the number of RF chains
relative to fully-digital arrays, but still requires distributed phase-control networks and
calibration, and its hardware complexity and power remain significant at large $N$
\cite{mendez2016hybrid}. Across these families, producing \emph{arbitrary} complex per-antenna
excitations generally requires explicit amplitude-control elements or lossy magnitude shaping.

A different approach to reducing RF-chain complexity is to exploit array physics and loading. In
beamspace/parasitic-array architectures, symbols are mapped to a small set of basis radiation
patterns \cite{alrabadi2012singleRF}. In load-modulated arrays, tunable reactive loads shape currents
and beams with a single RF source \cite{sedaghat2016load,hong2015loadmodulated}. These methods can be
effective when the antenna coupling and the load network are engineered and calibrated accurately,
but they typically entangle signal synthesis with mutual coupling, matching, and power-amplifier
loading, which complicates calibration and can limit robustness.

In integrated photonics, universal interferometric meshes implement arbitrary unitary transforms
using networks of $2\times2$ building blocks \cite{reck1994experimental,clements2016optimal,carolan2015universal},
with extensive work on programmability and calibration
\cite{miller2013selfconfiguring,bogaerts2020programmable,harris2018linear,perez2018fieldprogrammable}.

In the RF/microwave domain, \cite{keshavarz2025programmable} demonstrated a programmable unitary
processor using an interlaced architecture that targets realization of an arbitrary $N\times N$
unitary transformation, which typically requires $O(N^2)$ tunable elements and is programmed via
iterative optimization with $O(N^3)$ computational complexity.
Our paper targets a different synthesis objective: rather than realizing an arbitrary
$N\times N$ unitary, we only need to synthesize a \emph{single} desired excitation vector from a
single-tone input. This restriction enables a substantially lighter programmable architecture with
only $2N-1$ tunable degrees of freedom (matching the dimension of the complex unit sphere), together
with a closed-form programming rule that avoids iterative optimization and is therefore compatible
with low-latency, symbol-wise reconfiguration.

At a high level, the proposed magnitude-splitting stage resembles a \emph{corporate-feed} divider: both distribute
power from a single source to $N$ antenna ports through a binary splitting topology
(e.g., corporate-feed tapers in phased arrays~\cite{mailloux2005phasedarray,huang2015amplitudetapering,karimi2022amplitudetapered}).
However, the \emph{objective and operating regime} are fundamentally different.
Corporate-feed tapers are typically designed to realize a \emph{fixed} (or slowly reconfigurable) \emph{real,
nonnegative} amplitude law across the aperture (often with symmetry constraints) to shape the array factor and
reduce sidelobes, while maintaining (approximately) equal output phase for a chosen pointing direction~\cite{mailloux2005phasedarray,huang2015amplitudetapering,karimi2022amplitudetapered}.
In contrast, this paper targets \emph{universal, symbol-wise} synthesis of an \emph{arbitrary complex} excitation
vector $\bm x[m]\in\mathbb{C}^N$ under a fixed total-power constraint $\|\bm x[m]\|_2^2=P$.
That is, we must realize \emph{any} unit-norm direction $\bm c[m]=\bm x[m]/\sqrt{P}$ (including arbitrary per-antenna
phases) from a \emph{single coherent tone} by reprogramming the passive network each symbol interval.
Moreover, classical taper networks treat splitter-induced phase offsets and path-length differences as
impairments to be physically equalized~\cite{mailloux2005phasedarray,karimi2022amplitudetapered}; here we instead
exploit an interferometric splitting convention with \emph{known deterministic} branch phase offsets inside the
magnitude tree and compensate them explicitly using a dedicated per-antenna output phase bank.

A second closely related family is \emph{reconfigurable} divider trees and variable power-distribution networks
used for antenna pattern reconfiguration (e.g., switched or multimode beamforming networks)~\cite{angeletti2014multimode,fan2016switched}.
While these architectures can also employ tunable split ratios, they are commonly specified for (i) a \emph{finite}
set of routing/power states~\cite{kim2004reconfigurable,tae2012reconfigurable,alazemi2019fiveport}, or (ii) relatively slow beam/pattern adaptation, and their control is typically
expressed in terms of beam or coverage objectives rather than \emph{exact} realization of an arbitrary
$\bm c[m]\in\mathbb{C}^N$ at symbol cadence.\footnote{When continuous tuning is available, control is often implemented
via bias-controlled tuning and pre-characterized settings (look-up tables) and/or objective-driven adjustment of the tuning variables~\cite{hu2016compact,alazemi2019fiveport},
rather than a closed-form, parameter-minimal state-preparation rule.}
By contrast, this paper restricts attention to the \emph{single-input fanout} operating mode in which only one
network input is driven and the remaining $N\!-\!1$ inputs are impedance matched.
This reduced but practically important objective---\emph{state preparation} $\bm V(\bm\varphi[m])\bm e_1=\bm c[m]$ from
one driven port---enables a substantially lighter and more direct synthesis mechanism than general-purpose
reconfigurable BFNs.

Compared with conventional multi-antenna transmitter/front-end architectures, the proposed fully-analog
transmitter addresses the symbol-by-symbol synthesis of an arbitrary complex excitation vector from a single
coherent RF tone. A fully-digital array can generate an arbitrary
$\bm x[m]\in\mathbb{C}^N$ directly in the digital/RF-chain domain by using $N$ independent RF chains, but this
flexibility comes with RF-chain count, power consumption, synchronization/calibration burden, and hardware cost
that scale with $N$ \cite{mendez2016hybrid}. Hybrid analog/digital architectures reduce the number of RF chains
to $N_{\mathrm{RF}}<N$, but the analog stage is commonly constrained by phase-shifter or switch networks, so exact
symbol-by-symbol synthesis of an arbitrary complex aperture vector is limited, and is also not the native objective
\cite{mendez2016hybrid}. Conventional corporate-feed and tapering networks can also operate from a single RF
source \cite{mailloux2005phasedarray,huang2015amplitudetapering,karimi2022amplitudetapered}, and reconfigurable
divider-tree variants can introduce additional switching or tuning states
\cite{angeletti2014multimode,fan2016switched,kim2004reconfigurable,tae2012reconfigurable,alazemi2019fiveport,hu2016compact};
however, such networks are typically designed for fixed, slowly reconfigurable, finite-state, or prescribed
beam/pattern objectives rather than for exact symbol-by-symbol synthesis of an arbitrary complex excitation
vector. General programmable unitary meshes address the more general task of realizing an arbitrary
$N\times N$ unitary transformation and therefore require $O(N^2)$ tunable elements and, in many implementations,
iterative decomposition, optimization, or calibration procedures
\cite{reck1994experimental,clements2016optimal,carolan2015universal,miller2013selfconfiguring,bogaerts2020programmable,harris2018linear,perez2018fieldprogrammable,keshavarz2025programmable}.
In contrast, the architecture proposed in this paper does not attempt to synthesize an arbitrary $N\times N$
unitary matrix. It synthesizes only the single $N$-dimensional output vector required by the transmitter, namely
the desired antenna excitation vector $\bm x[m]$ from a single coherent RF tone. This  objective admits achievability via a binary tree of lossless variable splitters followed by output phase shifters, requiring only
$N-1$ magnitude-splitting controls and $N$ output phase controls, for a total of $2N-1$ tunable real quantities.
Moreover, the required settings are obtained by a deterministic closed-form programming rule with $O(N)$
arithmetic complexity, rather than by a general-purpose $N\times N$ unitary synthesis or iterative optimization
procedure.

The present contribution is therefore not merely a reconfigurable corporate-feed network, but a
\emph{parameter-minimal, closed-form programmable} unitary state-preparation architecture that emulates a
narrowband fully-digital $N$-antenna transmitter under a fixed total-power constraint.

\subsection{Contribution}
\label{subsec:contribution}

The present paper develops a fully-analog transmitter architecture and a constructive programming
rule for arbitrary complex $N$-antenna excitation synthesis from a single coherent tone under an
ideal unitary network model. Our main contributions are as follows:
\begin{itemize}
\item A fully-analog narrowband multi-antenna transmitter architecture that emulates the functionality of a narrowband fully-digital multi-antenna transmitter with a fixed total-power constraint: in symbol interval $m$, synthesize arbitrary complex $N$-antenna excitation vectors $\bm x[m]\in\mathbb{C}^N$ with $\|\bm x[m]\|_2^2=P$ from a single coherent RF tone injected into one port of a passive multiport network with the remaining inputs matched.

\item Since the architecture can synthesize any
$\bm x[m]\in\mathbb{C}^N$ at symbol cadence (subject to $\|\bm x[m]\|_2^2=P$), it can realize the same narrowband
linear-precoding/MU-MIMO transmit vectors $\bm x[m]=\bm W\,\bm d[m]$ as a fully-digital multi-antenna transmitter on a single
carrier or OFDM subcarrier, while using only one RF chain.

\item A unitary state-preparation formulation $\bm V(\bm\varphi[m])\bm e_1=\bm c[m]$ with $\bm c[m]=\bm x[m]/\sqrt{P}$ and $\|\bm c[m]\|_2=1$, implying that any antenna-port excitation $\bm x[m]$ satisfying $\|\bm x[m]\|_2^2=P$ can be realized by programming a lossless passive network in the ideal model.

\item A parameter-minimal split-then-phase realization: a balanced binary magnitude-splitting tree allocates the desired per-antenna magnitudes using $N-1$ tunable split ratios, and a per-antenna output phase bank assigns the desired phases using $N$ tunable phase shifts, for a total of $2N-1$ real tunable degrees of freedom.

\item A deterministic closed-form $O(N)$ programming rule (no iterative optimization) that computes the internal split settings from subtree norms and computes the output phase-bank settings by compensating the deterministic (or calibrated) phase offsets introduced by the magnitude tree. This supports symbol-wise updates when the phase-control technology supports the required tuning speed.

\item An idealized passive hardware realization as a depth-$L=\log_2 N$ balanced binary tree of interferometric $2\times 2$ fanout cells (Mach--Zehnder/hybrid-based) followed by a per-antenna output phase-shifter bank, providing a fully-analog surrogate for a fully-digital multi-antenna transmitter: one coherent source/power amplifier (PA) and a programmable passive distributor in place of $N$ RF chains.

\item A compute-excluded RF-front-end power model anchored to representative   commercial off-the-shelf (COTS) components, and a comparison against an equivalent COTS fully-digital array under a common delivered antenna-port power normalization. The numerical results for $N\le 16$ indicate significant RF-front-end DC-power savings for the fully-analog architecture.
\end{itemize}

While we focus on multi-antenna excitation synthesis for wireless transmission, the same phase-only
unitary synthesis principle applies to other coherent multi-output tone generation settings,
including coherent test/measurement distribution networks and array excitation for radar/sonar/
ultrasound systems.

\subsection{Scope and limitations}
\label{subsec:scope_limitations}

This paper is a proof-of-concept study that establishes an exact phase-only synthesis principle and a minimal
parameterization for generating an arbitrary transmit vector direction described by $\bm c[m]\in\mathbb{C}^N$ with
$\|\bm c[m]\|_2=1$ in symbol interval $m$ from a single coherent tone using an ideal unitary passive network.
Accordingly, we adopt a narrowband phasor model and ideal component assumptions to isolate the core synthesis
mechanism and its algorithmic/hardware mapping.

We explicitly do \emph{not} attempt to validate a fabricated RF prototype or to exhaustively characterize
implementation non-idealities. In particular, topics such as finite phase quantization, tuning-dependent loss and
amplitude/phase imbalance, non-unitary transfer due to imperfect matching and coupler errors, calibration
procedures and LUT construction, PA--load interaction, mutual coupling at the antenna ports, and spectral
regrowth due to control transients are important but are outside the scope of the present proof-of-concept and
are left for future hardware-focused work.

Consequently, the validation provided in this paper is analytical and algorithmic: exact synthesis is established
by Theorem~\ref{thm:exact_state_preparation} under the ideal unitary model, and the closed-form programming rule
is verified by the constructive derivation in Section~\ref{sec:exact_synthesis}. Fabricated-hardware validation,
device-specific calibration, and over-the-air measurements are reserved for a follow-up dedicated experimental
implementation study.

\subsection{Paper organization}
\label{subsec:paper_organization}

\Cref{sec:system_model} introduces the phasor model, the fully-analog transmitter architecture, and
the synthesis objective. \Cref{sec:exact_synthesis} develops the split-then-phase unitary
state-preparation construction and presents the closed-form programming rule for the magnitude tree
and output phase bank. \Cref{sec:hardware} maps the construction to an idealized passive hardware
realization as a binary-tree interferometric network followed by per-antenna phase control.
\Cref{sec:nonideal_framework} introduces a passive contractive non-ideal model and establishes the
error metrics and delivered-power normalization used for efficiency comparisons.
\Cref{sec:numerical_results} reports numerical results and practical implications. Bold lowercase
letters denote vectors and bold uppercase letters denote matrices; $(\cdot)^T$ and $(\cdot)^H$ denote
transpose and conjugate transpose, and $\|\cdot\|_2$ is the Euclidean norm.

 % ============================================================
% Section 2: System Model and Problem Formulation
% ============================================================
\section{System model and problem formulation}
\label{sec:system_model}

\subsection{Narrowband transmitter model}
\label{subsec:phasor_model}

We consider a narrowband $N$-antenna wireless transmitter operating at carrier frequency $f_c$. Over
symbol interval $m$ of duration $\Ts$, the $n$th antenna radiates the real RF waveform
\begin{align}
&x_n(t) = \sqrt{2}\Re\!\left\{ x_n[m]\,e^{j 2\pi f_c t}\right\},\\
&
\textrm{for }t\in[m\Ts,(m+1)\Ts),\ \ n=1,\dots,N,\nonumber
\end{align}
where $x_n[m]\in\mathbb{C}$ is the complex phasor that determines the amplitude $|x_n[m]|$ and phase
$\angle x_n[m]$ at antenna $n$ during symbol interval $m$. Collecting all antenna excitations yields
\begin{equation}
\bm x[m] \triangleq [x_1[m],\dots,x_N[m]]^T \in \mathbb{C}^N.
\end{equation}
Under the standard narrowband interpretation, the total radiated power\footnote{Under idealized conditions for antenna efficiencies and coupling. In non-ideal settings, $\|\bm x[m]\|_2^2$ is the total delivered (conducted) power at the antenna feed ports during symbol interval $m$.} over symbol interval $m$ is
$\|\bm x[m]\|_2^2$ and the power at antenna $n$ is $|x_n[m]|^2$. Note that the symbol rate is proportional to $1/\Ts$.

The phasor $x_n[m]$ can be interpreted as the (narrowband) complex-envelope sample driving antenna
$n$ over symbol interval $m$. In wireless communications, $\bm x[m]$ is produced jointly by a
baseband precoder and data symbols. The baseband precoder stays constant over many symbol intervals,
whereas the data symbols change in each symbol interval. 

For example, in a narrowband $S$-stream (or multi-user) MIMO transmission, $\bm d[m]\in\mathbb{C}^{S}$ collects
the $S$ data streams/users and $\bm W\in\mathbb{C}^{N\times S}$ is a linear precoder/beamformer, so that the
desired antenna vector is $\bm x[m]=\bm W\,\bm d[m]$ (with the usual per-symbol normalization chosen so that
$\|\bm x[m]\|_2^2=P$ holds).
Because the proposed fully-analog transmitter can synthesize \emph{any} $\bm x[m]$ satisfying
$\|\bm x[m]\|_2^2=P$, it can implement this MIMO mapping \emph{identically to a fully-digital array} on the
considered narrowband carrier (or per OFDM subcarrier), but with only a single RF chain and a passive
power-redistribution network.

The aim of this paper is to propose a narrowband fully-analog transmitter that, in symbol interval
$m$, synthesizes any desired $\bm x[m]$, such that $\|\bm x[m]\|_2^2=P$ holds. Thereby, the task of
synthesizing $\bm x[m]$ as $\bm x[m]=\bm W\,\bm d[m]$, such that $\|\bm x[m]\|_2^2=P$ holds, is also
achieved.

All synthesis and hardware derivations in this paper apply \emph{symbol-wise}. To keep notation
compact, we will often omit the explicit symbol index $[m]$ when no ambiguity arises.

\subsection{Fully-analog transmitter architecture}
\label{subsec:arch_model}

The transmitter is driven by a single coherent RF tone of power $P$,
\begin{equation}
s(t)=\sqrt{2}\Re\!\left\{ \sqrt{P}\,e^{j 2\pi f_c t}\right\}.
\end{equation}
This tone is injected into one input port of an $N\times N$ (i.e., $N$-input/$N$-output) passive
interferometric network that contains only fixed passive couplers/splitters and tunable
phase-control elements.\footnote{In the proposed realization, tunable phase control is used in two roles:
(i) to set interferometric \emph{split ratios} inside a magnitude-splitting tree (via differential phase in
$2\times2$ fanout cells), and (ii) to apply \emph{per-antenna} output phase shifts in a diagonal phase bank.}
The remaining $N-1$ input ports are terminated in matched loads, so their incident waves are zero.

We represent the programmable network by a narrowband complex matrix
$\bm V(\bm\varphi)\in\mathbb{C}^{N\times N}$ parameterized by tunable settings
$\bm\varphi\in\mathbb{R}^{M}$. The matrix $\bm V(\bm\varphi)$ maps the complex wave (phasor) vector at
the $N$ input reference planes to the complex wave (phasor) vector at the $N$ output reference
planes. We adopt standard power-wave normalization, so that (under matched and lossless conditions)
power preservation corresponds to a unitary mapping.

Under single-port excitation, let $\bm e_1\triangleq[1,0,\dots,0]^T$ denote the first canonical basis
vector. With matched unused inputs, the incident-wave vector at the network input reference planes is
\begin{equation}
\bm x_{\mathrm{in}} \triangleq \sqrt{P}\,\bm e_1\in\mathbb{C}^N.
\label{eq:single_port_excitation}
\end{equation}
The resulting antenna excitation vector at the network output reference planes is
\begin{equation}
\bm x = \bm V(\bm\varphi)\,\bm x_{\mathrm{in}}.
\label{eq:overall_mapping}
\end{equation}
No programmable attenuators and no active gain elements are assumed inside the synthesis network.

In this paper we construct $\bm V(\bm\varphi)$ in a parameter-minimal \emph{split-then-phase} form,
\begin{align}
&\bm V(\bm\varphi)
\;=\;
\underbrace{\bm P(\bm\vartheta)}_{\text{output phase bank}}\;
\underbrace{\bm V_{\mathrm{mag}}(\bm\alpha)}_{\text{magnitude tree}},\label{eq:split_then_phase_factorization}
\\
&\bm P(\bm\vartheta)\triangleq \mathrm{diag}\!\big(e^{j\vartheta_1},\dots,e^{j\vartheta_N}\big),\nonumber
\end{align}
where $\bm V_{\mathrm{mag}}(\bm\alpha)$ is a balanced binary-tree magnitude-splitting network controlled by
$N-1$ real parameters $\bm\alpha$ (split ratios) and $\bm P(\bm\vartheta)$ is a diagonal output phase bank with
$N$ tunable phases $\bm\vartheta=[\vartheta_1,\dots,\vartheta_N]^T$. Thus $M=(N-1)+N=2N-1$ real tunable degrees
of freedom are used, which matches the real dimension of the complex unit sphere and is therefore
parameter-minimal for universal single-vector synthesis. The concrete tree parameterization and the associated
closed-form programming rule are developed in \Cref{sec:exact_synthesis}.

\subsection{Ideal lossless model and unitary constraint}
\label{subsec:unitary_model}

In the idealized setting of perfectly matched, reciprocal, and lossless passive components (and
lossless phase-control elements), the network is power-preserving. Under the adopted normalization,
this corresponds to the unitary constraint
\begin{equation}
\bm V(\bm\varphi)^H \bm V(\bm\varphi) = \bm I_N,\qquad \text{(ideal lossless model)}.
\label{eq:unitary_constraint}
\end{equation}
Consequently, the network preserves the $\ell_2$-norm and therefore preserves total power:
\begin{equation}
\|\bm x\|_2^2
=
\|\bm V(\bm\varphi)\bm x_{\mathrm{in}}\|_2^2
=
\|\bm x_{\mathrm{in}}\|_2^2
=
|\sqrt{P}|^2
=
P.
\label{eq:power_preservation}
\end{equation}

\subsection{Design objective: prescribed transmit vector with prescribed total power}
\label{subsec:design_objective}

Given a desired transmit vector $\bm x\in\mathbb{C}^N$ with prescribed total power
\begin{equation}
\|\bm x\|_2^2 = P,\qquad P>0,
\label{eq:power_constraint}
\end{equation}
our goal is to propose a hardware-implementable unitary matrix family $\bm V(\bm\varphi)$ and to
determine the tunable settings $\bm\varphi$ such that the fully-analog transmitter \eqref{eq:overall_mapping}
generates $\bm x$ from the single coherent input tone.

Defining the unit-norm vector direction
\begin{equation}
\bm c \triangleq \frac{\bm x}{\sqrt{P}}\in\mathbb{C}^N,\qquad \|\bm c\|_2=1,
\label{eq:normalized_target}
\end{equation}
the synthesis task reduces to the following state-preparation problem:
\begin{equation}
\text{given }\bm c\in\mathbb{C}^N \text{ with }\|\bm c\|_2=1,\ \text{find }\bm\varphi\text{ s.t.: }
\bm V(\bm\varphi)\,\bm e_1 = \bm c.
\label{eq:state_preparation_problem}
\end{equation}
Under the split-then-phase structure \eqref{eq:split_then_phase_factorization}, this problem becomes: choose
the split settings $\bm\alpha$ so that $\bm V_{\mathrm{mag}}(\bm\alpha)\bm e_1$ has magnitudes $|c_n|$, and then
choose $\bm\vartheta$ so that the diagonal phase bank produces the desired phases $\angle c_n$.

\begin{remark}
For clarity of the constructive development and the associated binary-tree hardware mapping, in
\Cref{sec:exact_synthesis,sec:hardware} we focus on the case where $N$ is a power of two, i.e.,
$N=2^L$ for some integer $L\ge 1$. Handling non-power-of-two $N$ is addressed in
Section~\ref{subsec:general_N_padding} via zero-padding and optional tree pruning.
\end{remark}

\begin{remark}
If the injected tone has an unknown/common phase offset $e^{j\theta_s}$ relative to the calibration reference
(e.g., due to LO phase), then the ideal lossless network produces $e^{j\theta_s}\bm x$, i.e., the desired vector
up to a common phase rotation. Exact reproduction of $\bm x$ is recovered by subtracting $\theta_s$ from all
programmed output phases in $\bm P(\bm\vartheta)$ (equivalently, by programming the network for the rotated
direction $e^{-j\theta_s}\bm c$); see \Cref{sec:exact_synthesis}.
\end{remark}

% ============================================================
% Section 3: Exact Phase-Only Synthesis via Unitary State Preparation
% ============================================================
\section{Exact phase-only synthesis via unitary state preparation}
\label{sec:exact_synthesis}

This section derives a closed-form programming rule for generating a desired transmit vector
$\bm x\in\mathbb{C}^N$ from the single-port excitation $\bm x_{\mathrm{in}}=\sqrt{P}\bm e_1$.
Equivalently, we seek a lossless (unitary) matrix family $\bm V(\bm\varphi)$ such that
\begin{equation}
\bm V(\bm\varphi)\bm e_1=\bm c,\qquad \bm c\triangleq \bm x/\sqrt{P},\ \ \|\bm c\|_2=1.
\label{eq:state_preparation_problem_sec3}
\end{equation}
We use the split-then-phase factorization introduced in \Cref{sec:system_model}:
a balanced binary \emph{magnitude tree} $\bm V_{\mathrm{mag}}(\bm\alpha)$ allocates the desired per-antenna
magnitudes $|c_n|$, and a diagonal \emph{output phase bank} $\bm P(\bm\vartheta)$ assigns the desired phases
$\angle c_n$.

We first develop the construction for the balanced case in which the number of antennas is a power of two,
\begin{equation}
N=2^L,\qquad L\in\mathbb{Z}_{\ge 1}.
\label{eq:N_power_of_two}
\end{equation}
This case gives the cleanest binary-tree notation and enables a balanced realization with depth $L=\log_2 N$
and a closed-form $O(N)$ programming procedure. The general non-power-of-two case is then handled by a concise
zero-padding/pruning extension in Section~\ref{subsec:general_N_padding}, without changing the main programming
principle.

% ------------------------------------------------------------
\subsection{Primitive $2\times 2$ fanout splitter with fixed relative phase}
\label{subsec:splitter_cell}

The magnitude tree is built from a tunable $2\times 2$ interferometric splitter whose split ratio is controlled
by a single real parameter.

\begin{definition}[$2\times 2$ splitter cell]
\label{def:splitter}
Define the unitary $2\times 2$ splitter
\begin{equation}
\bm U(\alpha)\triangleq
\begin{bmatrix}
\cos\alpha & j\sin\alpha\\
j\sin\alpha & \cos\alpha
\end{bmatrix},\qquad \alpha\in[0,\pi/2].
\label{eq:Ualpha}
\end{equation}
\end{definition}

\begin{lemma}[Lossless splitting on an input $(a,0)$ with a known branch phase offset]
\label{lem:splitting}
The matrix $\bm U(\alpha)$ is unitary. Moreover, for any $a\in\mathbb{C}$,
\begin{equation}
\bm U(\alpha)\begin{bmatrix}a\\0\end{bmatrix}
=
\begin{bmatrix}
a\cos\alpha\\
j\,a\sin\alpha
\end{bmatrix}.
\label{eq:split}
\end{equation}
Thus $\alpha$ controls the magnitude split while the relative phase between the two outputs is a \emph{fixed}
known offset (here $+\pi/2$ on the ``$\sin$'' branch under the convention of \eqref{eq:Ualpha}), and power is
preserved: $|a|^2=|a\cos\alpha|^2+|a\sin\alpha|^2$.
\end{lemma}
\begin{IEEEproof}
A direct computation gives $\bm U(\alpha)^H\bm U(\alpha)=\bm I_2$. The relation \eqref{eq:split} follows by
multiplication.
\end{IEEEproof}

% ------------------------------------------------------------
\subsection{Balanced binary-tree magnitude network}
\label{subsec:magnitude_tree}

We now construct an $N\times N$ unitary $\bm V_{\mathrm{mag}}(\bm\alpha)$ that allocates magnitudes using only
split parameters $\{\alpha_{\ell,i}\}$ arranged on a balanced binary tree.

\subsubsection{Tree indexing (nodes, leaves, and contiguous leaf sets)}

We index internal nodes by levels $\ell\in\{1,2,\dots,L\}$, where $\ell=1$ is the root and $\ell=L$ is
the last internal level. Leaves correspond to antenna indices and are placed conceptually at level
$\ell=L+1$, indexed by $n\in\{1,\dots,N\}$.

At level $\ell$ there are $2^{\ell-1}$ nodes indexed by $i\in\{1,\dots,2^{\ell-1}\}$. Node $(\ell,i)$
covers the contiguous leaf set
\begin{equation}
S(\ell,i)\triangleq\Big\{(i-1)2^{L-\ell+1}+1,\dots,i2^{L-\ell+1}\Big\},
\label{eq:bt_S_def}
\end{equation}
so $S(1,1)=\{1,\dots,N\}$ and $S(L,i)=\{2i-1,2i\}$.

For $\ell\in\{1,\dots,L-1\}$ the children of node $(\ell,i)$ are
\begin{equation}
(\ell,i)_{\mathrm{Left}}=(\ell+1,2i-1),\qquad
(\ell,i)_{\mathrm{Right}}=(\ell+1,2i).
\label{eq:bt_children}
\end{equation}

\subsubsection{Embedding a $2\times 2$ splitter into an $N\times N$ matrix}

For indices $1\le m<n\le N$, let $\bm U_{m,n}(\alpha)\in\mathbb{C}^{N\times N}$ denote the identity matrix
$\bm I_N$ except on rows/columns $\{m,n\}$, where the $2\times2$ principal submatrix is replaced by
$\bm U(\alpha)$ in \eqref{eq:Ualpha}. Thus $\bm U_{m,n}(\alpha)$ is unitary and acts nontrivially only on
coordinates $m$ and $n$.

\subsubsection{Which coordinate pairs are mixed at level $\ell$?}

Define the \emph{entry index} (first leaf of the block) of node $(\ell,i)$ as
\begin{equation}
p(\ell,i)\triangleq (i-1)2^{L-\ell+1}+1,
\label{eq:bt_root_index}
\end{equation}
and define the corresponding \emph{split partner index} as
\begin{equation}
q(\ell,i)\triangleq p(\ell,i)+2^{L-\ell}.
\label{eq:bt_pair_indices}
\end{equation}
At level $\ell$, node $(\ell,i)$ splits the complex amplitude carried by coordinate $p(\ell,i)$ into two
descendants supported on indices
\[
p(\ell+1,2i-1)=p(\ell,i),\qquad p(\ell+1,2i)=q(\ell,i).
\]
Under our convention \eqref{eq:split}, the ``$\cos$'' branch propagates to the left child and the ``$\sin$''
branch (with fixed factor $j$) propagates to the right child.

\subsubsection{Layer matrices $\bm U_\ell$ and magnitude-tree unitary}

For each level $\ell\in\{1,\dots,L\}$ define the layer matrix
\begin{equation}
\bm U_\ell
\triangleq
\prod_{i=1}^{2^{\ell-1}}
\bm U_{\,p(\ell,i),\,q(\ell,i)}(\alpha_{\ell,i}),
\qquad \ell=1,\dots,L.
\label{eq:mag_layer_matrix}
\end{equation}
The index pairs $\big(p(\ell,i),q(\ell,i)\big)$ are disjoint across $i$ at fixed $\ell$, hence the factors in
\eqref{eq:mag_layer_matrix} commute and each $\bm U_\ell$ is unitary as a product of unitaries acting on disjoint
two-dimensional subspaces.

Define the magnitude-tree unitary as
\begin{equation}
\bm V_{\mathrm{mag}}(\bm\alpha)
\triangleq
\bm U_L \bm U_{L-1}\cdots \bm U_1.
\label{eq:Vmag_def}
\end{equation}
Since $\bm V_{\mathrm{mag}}(\bm\alpha)$ is a product of unitary matrices, it is unitary for all parameter values.

\subsubsection{Why every embedded $2\times2$ splitter sees one zero input}
\label{subsubsec:fanout_zero_input}

Define intermediate vectors
\begin{equation}
\bm v^{(1)} \triangleq \bm e_1,\qquad
\bm v^{(\ell+1)} \triangleq \bm U_\ell \bm v^{(\ell)},\ \ \ell=1,\dots,L,
\label{eq:mag_intermediate_vectors}
\end{equation}
so that $\bm v^{(L+1)}=\bm V_{\mathrm{mag}}(\bm\alpha)\bm e_1$.

\begin{lemma}[Zero-input property at each level]
\label{lem:zero_input_property}
For every level $\ell\in\{1,\dots,L\}$ and every node index $i\in\{1,\dots,2^{\ell-1}\}$,
\begin{equation}
\big[\bm v^{(\ell)}\big]_{q(\ell,i)} = 0.
\label{eq:zero_input_property}
\end{equation}
Consequently, each embedded $2\times2$ splitter inside $\bm U_\ell$ acts on a two-dimensional input of the form
$\big([\bm v^{(\ell)}]_{p(\ell,i)},\ 0\big)^T$ and therefore performs the fanout splitting in
Lemma~\ref{lem:splitting}.
\end{lemma}

\begin{IEEEproof}
We prove by induction on $\ell$.

For $\ell=1$, we have $\bm v^{(1)}=\bm e_1$, hence $[\bm v^{(1)}]_n=0$ for all $n\neq 1$, in particular
$[\bm v^{(1)}]_{q(1,1)}=0$.

Assume \eqref{eq:zero_input_property} holds for some $\ell$. Consider $\bm v^{(\ell+1)}=\bm U_\ell \bm v^{(\ell)}$.
Each factor in $\bm U_\ell$ acts only on its own pair $\{p(\ell,i),q(\ell,i)\}$ and leaves all other coordinates
unchanged. By the induction hypothesis, the input on $q(\ell,i)$ is zero for every $i$, so the action of that
factor can only create nonzero output on $q(\ell,i)$ \emph{at level $\ell+1$}. By construction of the index sets,
each such $q(\ell,i)$ is an entry index at the next level, i.e., $q(\ell,i)=p(\ell+1,2i)$. Therefore, no
coordinate of the form $q(\ell+1,\cdot)$ has been created \emph{before} applying $\bm U_{\ell+1}$, and hence
$[\bm v^{(\ell+1)}]_{q(\ell+1,i)}=0$ for all $i$. This proves the claim for $\ell+1$.
\end{IEEEproof}

% ------------------------------------------------------------
\subsection{Exact split-then-phase state preparation in closed form}
\label{subsec:split_then_phase_exact}

We now choose the split parameters $\bm\alpha$ and the output phases $\bm\vartheta$ in closed form so that
\begin{equation}
\bm P(\bm\vartheta)\,\bm V_{\mathrm{mag}}(\bm\alpha)\,\bm e_1 \;=\; \bm c,
\label{eq:split_then_phase_goal}
\end{equation}
which (by \eqref{eq:overall_mapping} and \eqref{eq:single_port_excitation}) implies $\bm x=\sqrt{P}\bm c$ at the
antenna ports.

\subsubsection{Subtree norms (magnitude targets)}
\label{subsubsec:subtree_norms}

For each node $(\ell,i)$, define its subtree norm
\begin{equation}
r_{\ell,i}\triangleq \|\bm c_{S(\ell,i)}\|_2,\qquad \ell=1,\dots,L,\ \ i=1,\dots,2^{\ell-1}.
\label{eq:bt_r_def}
\end{equation}
Define leaf-level values at $\ell=L+1$ as
\begin{equation}
r_{L+1,n}\triangleq |c_n|,\qquad n=1,\dots,N.
\label{eq:bt_leaf_r_def}
\end{equation}
Then $r_{1,1}=1$ and the subtree norms satisfy the recursion
\begin{equation}
r_{\ell,i}^2 = r_{\ell+1,2i-1}^2 + r_{\ell+1,2i}^2,\qquad
\ell=1,\dots,L.
\label{eq:bt_energy_recursion}
\end{equation}

\subsubsection{Split programming rule (closed form)}
\label{subsubsec:split_programming}

Choose the split angles at each internal node as
\begin{equation}
\alpha_{\ell,i}\triangleq
\begin{cases}
0, & r_{\ell,i}=0,\\[2pt]
\operatorname{atan2}\!\big(r_{\ell+1,2i},\,r_{\ell+1,2i-1}\big)\in[0,\pi/2], & r_{\ell,i}>0,
\end{cases}
\label{eq:alpha_choice}
\end{equation}
where $\ell=1,\dots,L$ and  $i=1,\dots,2^{\ell-1}$.
With this choice, the fanout splitting at each node allocates power between its left and right subtrees in
proportion to the desired subtree norms.

\subsubsection{Deterministic magnitude-tree phase offsets}
\label{subsubsec:mag_phase_offsets}

Because \eqref{eq:split} introduces a fixed factor $j$ on the ``$\sin$'' (right) branch output, the phases
produced by $\bm V_{\mathrm{mag}}(\bm\alpha)\bm e_1$ are deterministic functions of the root-to-leaf paths.

\begin{definition}[Right-branch count $w_n$]
\label{def:wn}
Let $N=2^L$. For each leaf index $n\in\{1,\dots,N\}$, write
\begin{equation}
n-1=\sum_{k=0}^{L-1} b_k 2^k,\qquad b_k\in\{0,1\},
\end{equation}
and define
\begin{equation}
w_n \triangleq \sum_{k=0}^{L-1} b_k,
\label{eq:wn_def}
\end{equation}
i.e., $w_n$ is the number of ``right-child'' decisions along the root-to-leaf path to leaf $n$
(equivalently the Hamming weight of the $L$-bit representation of $n-1$).
\end{definition}

\begin{lemma}[Phase-offset decomposition of the magnitude-tree output]
\label{lem:mag_phase_decomp}
Let $\bm y\triangleq \bm V_{\mathrm{mag}}(\bm\alpha)\bm e_1$. Under the splitter convention in \eqref{eq:Ualpha},
the phase of $y_n$ satisfies
\begin{equation}
y_n = |y_n|\,e^{j\theta^{(\mathrm{mag})}_n},
\qquad
\theta^{(\mathrm{mag})}_n \equiv \frac{\pi}{2}w_n \ \ (\mathrm{mod}\ 2\pi),
\label{eq:theta_mag_ideal}
\end{equation}
and hence $\theta^{(\mathrm{mag})}_n$ is deterministic given the tree labeling. In practice, fixed
layout-dependent phase offsets can be absorbed by replacing \eqref{eq:theta_mag_ideal} with a calibrated
per-output vector $\{\theta^{(\mathrm{mag})}_n\}$.
\end{lemma}

\begin{IEEEproof}
By Lemma~\ref{lem:zero_input_property}, each embedded splitter is driven in fanout mode and the right output
introduces a fixed multiplicative factor $j$ relative to its input (Lemma~\ref{lem:splitting}). Along the unique
root-to-leaf path to leaf $n$, each right-branch decision contributes one factor $j$, and left branches contribute
no additional phase. Hence $y_n = j^{w_n}|y_n|$, which implies
$\theta^{(\mathrm{mag})}_n \equiv (\pi/2)w_n \ (\mathrm{mod}\ 2\pi)$.
\end{IEEEproof}

\subsubsection{Closed-form output phase-bank programming and exactness}
\label{subsubsec:phase_bank_programming}

\begin{theorem}[Exact split-then-phase state preparation]
\label{thm:exact_state_preparation}
Let $N=2^L$ and let $\bm c\in\mathbb{C}^N$ satisfy $\|\bm c\|_2=1$. Define leaf magnitudes
$r_{L+1,n}\triangleq |c_n|$ and compute subtree norms $r_{\ell,i}$ by the recursion
$r_{\ell,i}^2=r_{\ell+1,2i-1}^2+r_{\ell+1,2i}^2$ (cf.~\eqref{eq:bt_energy_recursion}). Choose split angles
$\{\alpha_{\ell,i}\}$ according to \eqref{eq:alpha_choice}. Let $\theta^{(\mathrm{mag})}_n$ be given by the ideal
closed form \eqref{eq:theta_mag_ideal} (or by a calibrated per-output offset vector), and choose the output
phase-bank settings
\begin{equation}
\vartheta_n \triangleq
\begin{cases}
0, & c_n=0\ \text{(arbitrary)},\\[2pt]
\angle c_n - \theta^{(\mathrm{mag})}_n, & c_n\neq 0,
\end{cases}
\label{eq:phasebank_choice}
\end{equation}
where $n=1,\dots,N$, and all angles are interpreted modulo $2\pi$. Then the split-then-phase network satisfies
\begin{equation}
\bm P(\bm\vartheta)\,\bm V_{\mathrm{mag}}(\bm\alpha)\,\bm e_1 \;=\; \bm c.
\end{equation}
\end{theorem}

\begin{IEEEproof}
Let $\bm y=\bm V_{\mathrm{mag}}(\bm\alpha)\bm e_1$.

\emph{Magnitudes:}
By Lemma~\ref{lem:zero_input_property}, every embedded splitter is driven in fanout mode with one zero input, so
Lemma~\ref{lem:splitting} applies at every internal node. With \eqref{eq:alpha_choice}, each node allocates power
between its left and right children in proportion to the desired subtree energies
$r_{\ell+1,2i-1}^2$ and $r_{\ell+1,2i}^2$. By recursion down the tree, the induced leaf magnitudes satisfy
$|y_n|=|c_n|$ for all $n$.

\emph{Phases:}
By Lemma~\ref{lem:mag_phase_decomp}, $y_n=|y_n|e^{j\theta^{(\mathrm{mag})}_n}$. Multiplying by the diagonal phase
bank yields
\[
[\bm P(\bm\vartheta)\bm y]_n
=
|c_n|e^{j(\theta^{(\mathrm{mag})}_n+\vartheta_n)}
=
|c_n|e^{j\angle c_n}
=
c_n
\]
for all $n$ with $c_n\neq 0$. When $c_n=0$, the output is zero regardless of $\vartheta_n$. Therefore,
$\bm P(\bm\vartheta)\bm V_{\mathrm{mag}}(\bm\alpha)\bm e_1=\bm c$.
\end{IEEEproof}

\begin{corollary}
\label{cor:exact_tx_vector}
Let $N=2^L$ and let $\bm x\in\mathbb{C}^N$ satisfy $\|\bm x\|_2^2=P$ and define $\bm c=\bm x/\sqrt{P}$.
Compute $\bm\alpha$ and $\bm\vartheta$ using \eqref{eq:alpha_choice} and \eqref{eq:phasebank_choice}.
Then, for the single-port excitation $\bm x_{\mathrm{in}}=\sqrt{P}\bm e_1$ in \eqref{eq:single_port_excitation},
the fully-analog transmitter \eqref{eq:overall_mapping} produces
\[
\bm V(\bm\varphi)\bm x_{\mathrm{in}}
=
\bm P(\bm\vartheta)\bm V_{\mathrm{mag}}(\bm\alpha)\,\sqrt{P}\bm e_1
=
\sqrt{P}\,\bm c
=
\bm x.
\]
\end{corollary}

\begin{IEEEproof}
By \Cref{thm:exact_state_preparation}, $\bm P(\bm\vartheta)\bm V_{\mathrm{mag}}(\bm\alpha)\bm e_1=\bm c$. Multiply
both sides by $\sqrt{P}$ and use $\bm x_{\mathrm{in}}=\sqrt{P}\bm e_1$.
\end{IEEEproof}

% ------------------------------------------------------------
\subsection{Closed-form programming algorithm}
\label{subsec:programming_algorithm}

For completeness, Algorithm~\ref{alg:closed_form_phase_programming} summarizes the closed-form routine implied
by \Cref{thm:exact_state_preparation}. It is intended to be executed once per desired symbol-wise target vector.

\begin{algorithm}[t]
\caption{Closed-form programming for split-then-phase state preparation (binary tree, $N=2^L$)}
\label{alg:closed_form_phase_programming}
\normalsize
\begin{algorithmic}[1]
\Require Target $\bm x\in\mathbb{C}^N$ with $P=\|\bm x\|_2^2>0$ and $N=2^L$
\Ensure Split angles $\{\alpha_{\ell,i}\}$ and output phases $\bm\vartheta=[\vartheta_1,\dots,\vartheta_N]^T$

\State $\bm c \gets \bm x/\sqrt{P}$; \quad $L \gets \log_2 N$ \Comment{$\|\bm c\|_2=1$}
\State Leaf values ($\ell=L+1$): for all $n$, set $r_{L+1,n}\gets |c_n|$ and $\phi_n\gets \angle c_n$
\Comment{$\angle 0 \triangleq 0$}

\For{$\ell \gets L$ down to $1$} \Comment{Bottom-up magnitudes}
  \For{$i \gets 1$ to $2^{\ell-1}$}
    \State $r_{\ell,i}\gets \sqrt{r_{\ell+1,2i-1}^2+r_{\ell+1,2i}^2}$
    \State $\alpha_{\ell,i}\gets \operatorname{atan2}\!\big(r_{\ell+1,2i},\,r_{\ell+1,2i-1}\big)$; \quad (if $r_{\ell,i}=0$ set $\alpha_{\ell,i}\gets 0$)
  \EndFor
\EndFor

\State Compute known magnitude-tree phase offsets $\theta^{(\mathrm{mag})}_n$:
\Statex \quad Ideal convention: $\theta^{(\mathrm{mag})}_n \gets \frac{\pi}{2}w_n$ where $w_n$ is the Hamming weight of $n-1$ (equivalently the number of right-branch decisions).

\For{$n \gets 1$ to $N$}
  \State $\vartheta_n \gets 0$ if $c_n=0$, else $\vartheta_n \gets \phi_n - \theta^{(\mathrm{mag})}_n$
\EndFor

\State Program $\bm V_{\mathrm{mag}}(\bm\alpha)$ and $\bm P(\bm\vartheta)$.

\Statex \textbf{Optional:} If the injected tone has common phase offset $e^{j\theta_s}$, replace
$\vartheta_n \leftarrow \vartheta_n-\theta_s$ for all $n$.
\end{algorithmic}
\end{algorithm}

\begin{remark}[Degrees of freedom]
The split-then-phase architecture uses $N-1$ independent split controls $\{\alpha_{\ell,i}\}$ (one per internal
node) and $N$ output phase controls $\{\vartheta_n\}$, totaling $2N-1$ real tunable degrees of freedom
(parameter-minimal for universal state preparation).
\end{remark}

 % ------------------------------------------------------------
\subsection{General $N$ by zero-padding and optional pruning}
\label{subsec:general_N_padding}

The balanced-tree construction above is written for $N=2^L$ only to keep the indexing compact. For a general
number of antennas, define
\[
N' \triangleq 2^{\lceil \log_2 N\rceil}
\]
and form the padded target vector $\bm c'\in\mathbb{C}^{N'}$ as
\[
c'_n =
\begin{cases}
c_n, & 1\le n\le N,\\
0, & N<n\le N'.
\end{cases}
\]
Since $\|\bm c'\|_2=\|\bm c\|_2=1$, Theorem~\ref{thm:exact_state_preparation} applies directly to the padded
vector $\bm c'$. Therefore, Algorithm~\ref{alg:closed_form_phase_programming} can be applied with $N$ replaced
by $N'$. The first $N$ leaves realize the desired physical antenna excitations, while the padded leaves have
zero target amplitude. In an unpruned implementation, these padded outputs may simply be terminated.
Equivalently, all-zero padded subtrees may be removed, and any node with only one active child may be bypassed
or programmed with $\alpha=0$ or $\alpha=\pi/2$ to route all power to the active child. Hence, the general-$N$
case is handled without changing the closed-form programming rule; the balanced padded implementation uses
$N'-1<2N$ internal splitter cells before optional pruning.

% ------------------------------------------------------------
\subsection{Interpretation: magnitude splitting followed by per-antenna phase assignment}
\label{subsec:interpretation}

The construction admits an intuitive interpretation. Starting from the fixed input $\bm e_1$, each internal tree
node performs a lossless fanout split \eqref{eq:split} that allocates power between its left and right subtrees
according to \eqref{eq:alpha_choice}. After $L$ layers, the magnitude tree produces per-antenna magnitudes
$|c_1|,\dots,|c_N|$ while introducing deterministic phase offsets due to the fixed branch convention. The output
phase bank then compensates these offsets and assigns the desired phases $\angle c_n$, yielding exact synthesis
of $\bm c$ in the ideal lossless model. In \Cref{sec:hardware} we map this construction to an interferometric RF
network built from $2\times2$ fanout cells (e.g., MZI/hybrid-based) followed by per-antenna phase shifters.

% ------------------------------------------------------------
\subsection{Brief calibrated extension for non-unitary magnitude trees}
\label{subsec:calibrated_nonunitary_extension}

The closed-form rule above is exact for the ideal unitary splitter model. In a calibrated non-ideal
implementation, the local fanout response of an internal node $\nu$ can instead be represented as
\[
a \mapsto
\begin{bmatrix}
h_{\nu,\mathrm{L}}(u_\nu)a\\
h_{\nu,\mathrm{R}}(u_\nu)a
\end{bmatrix},
\]
where $u_\nu$ is the hardware control code and
$h_{\nu,\mathrm{L}}(u_\nu)$ and $h_{\nu,\mathrm{R}}(u_\nu)$ are calibrated complex branch gains that include
insertion loss, amplitude imbalance, and static phase offsets. The ideal rule
$\alpha_\nu=\operatorname{atan2}(r_{\nu,\mathrm{R}},r_{\nu,\mathrm{L}})$ can then be replaced by a LUT-based
branch-ratio selection. For nonzero requested child norms, one possible calibrated rule is
\[
u_\nu^\star
=
\arg\min_{u_\nu}
\left|
\frac{|h_{\nu,\mathrm{R}}(u_\nu)|^2G_{\nu,\mathrm{R}}}
     {|h_{\nu,\mathrm{L}}(u_\nu)|^2G_{\nu,\mathrm{L}}}
-
\frac{r_{\nu,\mathrm{R}}^2}{r_{\nu,\mathrm{L}}^2}
\right|,
\]
where $G_{\nu,\mathrm{L}}$ and $G_{\nu,\mathrm{R}}$ denote calibrated downstream power-transmission factors of
the left and right subtrees. If one requested child norm is zero, the limiting choice routes all available power
to the nonzero child. The accumulated calibrated branch phases are then absorbed into the output phase bank.
In the ideal case, $h_{\nu,\mathrm{L}}=\cos\alpha_\nu$, $h_{\nu,\mathrm{R}}=j\sin\alpha_\nu$, and
$G_{\nu,\mathrm{L}}=G_{\nu,\mathrm{R}}=1$, so the calibrated ratio rule reduces to
\eqref{eq:alpha_choice}. The realized direction error and delivered-power normalization are then evaluated using
the passive contractive model in Section~\ref{sec:nonideal_framework}.

% ============================================================
% Section 4: Idealized Analog Hardware Realization
% ============================================================
\section{Idealized analog hardware realization}
\label{sec:hardware}

This section maps the split-then-phase synthesis model of \Cref{sec:exact_synthesis} to an idealized
RF/microwave implementation built from transmission lines (or waveguides), fixed $3$\,dB couplers, and tunable
phase shifters. We view the programmable network as an $N$-input/$N$-output transmission network between an input reference plane
and an output reference plane, with output port $n$ feeding antenna element $n$.

At the input reference plane, only port~$1$ is driven by a coherent tone, while the remaining input ports are
matched so that their externally incident waves are zero. Under the ideal lossless and perfectly matched
assumptions adopted throughout this section, the network transfer is unitary and provides no RF gain.

Our objective is to realize in hardware the programmable unitary
\begin{equation}
\bm V(\bm\varphi)\;\triangleq\;\bm P(\bm\vartheta)\,\bm V_{\mathrm{mag}}(\bm\alpha),
\qquad
\bm\varphi \triangleq (\bm\alpha,\bm\vartheta),
\label{eq:V_split_then_phase_hw}
\end{equation}
so that, under single-port excitation $\bm x_{\mathrm{in}}=\sqrt{P}\bm e_1$,
\begin{equation}
\bm x
\;=\;
\bm V(\bm\varphi)\,\bm x_{\mathrm{in}}
\;=\;
\sqrt{P}\,\bm P(\bm\vartheta)\,\bm V_{\mathrm{mag}}(\bm\alpha)\,\bm e_1
\;=\;
\sqrt{P}\,\bm c,
\label{eq:hw_objective_mapping}
\end{equation}
where $\bm c$ is the desired unit-norm excitation direction (Theorem~\ref{thm:exact_state_preparation}).

% ------------------------------------------------------------
\subsection{Balanced magnitude-tree topology and correspondence to the matrix product}
\label{subsec:bt_topology_hw}

In \Cref{sec:exact_synthesis}, the magnitude-allocation network is constructed as the layered product
\begin{equation}
\bm V_{\mathrm{mag}}(\bm\alpha)
=
\bm U_L \bm U_{L-1}\cdots \bm U_1,
\qquad N=2^L,
\end{equation}
where each layer matrix $\bm U_\ell$ (\Cref{eq:mag_layer_matrix}) embeds $2^{\ell-1}$ disjoint tunable $2\times2$
splitters $\bm U(\alpha)$ (\Cref{eq:Ualpha}) into an $N\times N$ unitary.

A direct circuit interpretation is therefore a \emph{balanced binary tree} of $N-1$ tunable $2\times2$
interferometric splitter cells (one cell per internal tree node) that allocates the required magnitudes, followed
by a bank of $N$ per-antenna phase shifters implementing the diagonal matrix $\bm P(\bm\vartheta)$ in
\eqref{eq:V_split_then_phase_hw}. The tree has $L=\log_2 N$ layers of splitter cells:
layer $\ell\in\{1,\dots,L\}$ contains $2^{\ell-1}$ cells indexed by $i\in\{1,\dots,2^{\ell-1}\}$, matching the
indexing used in \Cref{sec:exact_synthesis}.

% ------------------------------------------------------------
\subsection{Interferometric realization of the $2\times 2$ splitter $\bm U(\alpha)$ with one input matched}
\label{subsec:2x2_cell}

Each internal node in the magnitude tree implements the $2\times2$ splitter $\bm U(\alpha)$ in \Cref{eq:Ualpha}.
A convenient realization is a Mach--Zehnder interferometer (MZI) formed by two fixed $3$\,dB couplers and one
tunable differential phase.

Let $\bm H$ denote the ideal $3$\,dB coupler model
\begin{equation}
\bm H \triangleq \frac{1}{\sqrt{2}}
\begin{bmatrix}
1 & 1\\
1 & -1
\end{bmatrix}.
\label{eq:hybrid_H}
\end{equation}
An MZI with internal differential phase $\delta$ has transfer matrix
\begin{equation}
\bm U_{\mathrm{MZI}}(\delta)
\triangleq
\bm H\,
\mathrm{diag}\!\left(e^{j\delta/2},\,e^{-j\delta/2}\right)
\bm H.
\label{eq:mzi_matrix}
\end{equation}
Using \eqref{eq:hybrid_H}, this simplifies to
\begin{equation}
\bm U_{\mathrm{MZI}}(\delta)
=
\begin{bmatrix}
\cos(\delta/2) & j\sin(\delta/2)\\
j\sin(\delta/2) & \cos(\delta/2)
\end{bmatrix}.
\label{eq:mzi_matrix_simplified}
\end{equation}

Comparing \eqref{eq:mzi_matrix_simplified} with \Cref{eq:Ualpha} shows that
\begin{equation}
\bm U(\alpha)=\bm U_{\mathrm{MZI}}(2\alpha),
\qquad \alpha\in[0,\pi/2],
\label{eq:Ualpha_as_MZI}
\end{equation}
so the split ratio is controlled by a single tunable differential phase $\delta=2\alpha$.

\begin{figure}[t]
\centering
\resizebox{\columnwidth}{!}{%
\begin{tikzpicture}[
    >={Stealth[length=2mm]},
    line width=0.8pt,
    font=\normalsize,
    diff/.style={draw, minimum width=22mm, minimum height=16mm, align=center, inner sep=2pt, fill=white},
    coup/.style={draw, minimum width=8mm, minimum height=14mm, align=center, fill=white},
    term/.style={draw, minimum width=8mm, minimum height=6mm, align=center, fill=white},
    node distance=1.0cm and 1.5cm,
]

\def\vsep{1.2}
\pgfmathsetmacro{\ymid}{0.5*\vsep}

\coordinate (in_top) at (0, \vsep);
\coordinate (in_bot) at (0, 0);

\node[anchor=east] at (in_top) {Input};

\node[term, anchor=east] (Z0) at (in_bot) {$Z_0$};

\node[coup] (c1) at (2.2, \ymid) {\scriptsize 3\,dB};
\node[coup] (c2) at (7.2, \ymid) {\scriptsize 3\,dB};

\node[diff] (dph) at (4.7, \ymid) {diff.\ phase\\$\pm \delta/2$};

\coordinate (out_top) at (9.4, \vsep);
\coordinate (out_bot) at (9.4, 0);

\node[anchor=west] at (out_top) {Output 1};
\node[anchor=west] at (out_bot) {Output 2};

\draw (in_top) -- (c1.west |- in_top);
\draw (Z0.east) -- (c1.west |- in_bot);

\draw (c1.east |- in_top) -- (dph.west |- in_top);
\draw (c1.east |- in_bot) -- (dph.west |- in_bot);

\draw (dph.east |- in_top) -- (c2.west |- in_top);
\draw (dph.east |- in_bot) -- (c2.west |- in_bot);

\draw (c2.east |- in_top) -- (out_top);
\draw (c2.east |- in_bot) -- (out_bot);

\end{tikzpicture}%
}
\caption{Fanout-mode interferometric splitter implementing $\bm U(\alpha)=\bm U_{\mathrm{MZI}}(2\alpha)$ with one
input matched and unexcited (incident wave $0$). Two fixed $3$\,dB couplers form an MZI whose differential
control $\delta=2\alpha$ sets the power split. The fixed $\pm 90^\circ$ branch convention of the MZI
(\Cref{eq:mzi_matrix}) is compensated by the per-antenna output phase bank $\bm P(\bm\vartheta)$.}
\label{fig:splitter_cell_matched_load}
\end{figure}

In the state-preparation operating mode, Lemma~\ref{lem:zero_input_property} implies that the second input to
each embedded $2\times2$ splitter is identically zero (no incident wave). In hardware, we enforce this by
terminating the corresponding idle input port in a matched load ($Z_0$), so reflections are suppressed and the
cell operates in the fanout regime of Lemma~\ref{lem:splitting}. Consequently, each internal node requires only
a \emph{single} tunable control $\delta_{\ell,i}=2\alpha_{\ell,i}$ to set the split ratio.

\begin{remark}
In the ideal model the terminated idle inputs carry (exactly) zero power in the fanout operating regime, so these
terminations do not dissipate signal power; under non-ideal isolation/coupler imbalance, leakage into the
terminations contributes an additional loss term that is absorbed into $g(\bm\eta)$ in
Section~\ref{sec:nonideal_framework}.
\end{remark}

% ------------------------------------------------------------
\subsection{Layered realization, physical pairing, and output phase bank}
\label{subsec:bt_layered_realization}

A convenient physical realization is to implement the product $\bm U_L\cdots \bm U_1$ as $L$ successive layers
of disjoint two-line splitter cells, where the pairing at layer $\ell$ follows the index map
$\big(p(\ell,i),q(\ell,i)\big)$ in \Cref{eq:bt_root_index,eq:bt_pair_indices}. Because the pairs are disjoint
within each layer, the corresponding cells can be placed in parallel at that layer (subject to routing
constraints), which directly realizes the commuting product structure in \Cref{eq:mag_layer_matrix}.

Under single-port excitation and ideal matching, the network operates in the fanout mode of
Lemma~\ref{lem:splitting} at every internal node, with split parameters $\{\alpha_{\ell,i}\}$ computed by the
closed-form rule in \Cref{sec:exact_synthesis} (Algorithm~\ref{alg:closed_form_phase_programming}). After the
magnitude tree, a bank of $N$ per-antenna phase shifters implements $\bm P(\bm\vartheta)$, with $\vartheta_n$
computed in closed form (Theorem~\ref{thm:exact_state_preparation}).

Fig.~\ref{fig:full_layout_tree_block} illustrates the proposed fully-analog transmitter architecture. The signal
flow from a single coherent RF tone to the desired antenna excitations $x_1,\ldots,x_N$ is as follows. First, note that  each block
$\mathrm{Cell}_{\ell,i}$ in the figure denotes one fanout-mode interferometric splitter of
Fig.~\ref{fig:splitter_cell_matched_load}. Thus, the label $\delta_{\ell,i}=2\alpha_{\ell,i}$ is the MZI
differential phase that sets the power split at node $(\ell,i)$. Next, 
the single RF tone first enters the root cell $\mathrm{Cell}_{1,1}$, where the differential phase
$\delta_{1,1}=2\alpha_{1,1}$ splits the incident power between the two child subtrees. More generally, at layer
$\ell$, the cell $\mathrm{Cell}_{\ell,i}$ receives the wave assigned to subtree $S(\ell,i)$ and splits it between
its left and right child subtrees according to the corresponding subtree norms of the normalized target vector
$\bm c=\bm x/\sqrt{P}$ used in Algorithm~\ref{alg:closed_form_phase_programming}. The second input of each
$2\times2$ cell is terminated in $Z_0$, which realizes the zero-incident-wave condition required for the fanout
mode in Lemma~\ref{lem:splitting}.
Finally, after $L$ splitting layers, the normalized leaf magnitudes are $|c_n|$ and, with the input tone power $P$, the
corresponding antenna-port magnitudes are $|x_n|=\sqrt{P}|c_n|$. The output phase bank $e^{j\vartheta_n}$ then
compensates the deterministic branch phases introduced by the splitter tree and assigns the final antenna phases,
thereby producing the desired antenna excitations $x_1,\ldots,x_N$.

\begin{figure*} 
\centering
\resizebox{\linewidth}{!}{%
\begin{tikzpicture}[
    % --- Styling ---
    >={Stealth[length=2mm]},
    font=\normalsize,
    % Component Styles
    port/.style={circle, draw, fill=white, inner sep=1.5pt, line width=0.8pt},
    src/.style={draw, rounded corners, align=center, fill=white, font=\normalsize, inner sep=6pt, line width=0.8pt},
    term/.style={draw, align=center, minimum height=6mm, minimum width=8mm, fill=white, font=\footnotesize},
    ph/.style={draw, align=center, minimum height=8mm, minimum width=14mm, fill=white},
    cellbox/.style={draw, rounded corners, fill=gray!5, align=center, minimum width=3.6cm, minimum height=1.4cm, line width=0.8pt},
    % Line Styles
    wire/.style={thick, ->, rounded corners=5pt},
    dotwire/.style={thick, densely dotted, ->, rounded corners=5pt}
]

% --- Layout Configuration ---
\def\levelDist{4.0}
\def\sibDist{9.0}
\def\leafOffset{2.5}
\def\srcX{0.9}

% --- Smart Cell Macro ---
% #1: Name, #2: X, #3: Y, #4: Label Content
\newcommand{\SmartCell}[4]{
    \node[cellbox] (#1) at (#2, #3) {#4};

    \coordinate (#1_inL)  at ([xshift=0.9cm]#1.north west);
    \coordinate (#1_inR)  at ([xshift=-0.9cm]#1.north east);
    \coordinate (#1_outL) at ([xshift=0.9cm]#1.south west);
    \coordinate (#1_outR) at ([xshift=-0.9cm]#1.south east);

    \node[port] at (#1_inL) {};
    \node[port] at (#1_inR) {};
    \node[port] at (#1_outL) {};
    \node[port] at (#1_outR) {};

    % Matched termination on the idle (left) input
    \node[term] (#1_Z0) at ([xshift=-1.1cm, yshift=0.5cm]#1.north west) {$Z_0$};
    \draw[thick, rounded corners=3pt] (#1_Z0.east) -| (#1_inL.north);
}

% ==========================================
% 1. Place representative cells (top, one intermediate layer, bottom)
% ==========================================

% Root (level 1)
\SmartCell{cRoot}{0}{0}{
    \textbf{Cell}$_{1,1}$\\[3pt]
    $\delta_{1,1}=2\alpha_{1,1}$
}

% Level 2 (two children of root)
\SmartCell{c21}{-\sibDist/2}{-\levelDist}{
    \textbf{Cell}$_{2,1}$\\[3pt]
    $\delta_{2,1}=2\alpha_{2,1}$
}
\SmartCell{c22}{\sibDist/2}{-\levelDist}{
    \textbf{Cell}$_{2,2}$\\[3pt]
    $\delta_{2,2}=2\alpha_{2,2}$
}

% Bottom (representative last layer, level L)
\SmartCell{cL1}{-\sibDist/2 - \leafOffset}{-3.0*\levelDist}{
    \textbf{Cell}$_{L,1}$\\[3pt]
    $\delta_{L,1}=2\alpha_{L,1}$
}
\SmartCell{cLR}{\sibDist/2 + \leafOffset}{-3.0*\levelDist}{
    \textbf{Cell}$_{L,2^{L-1}}$\\[3pt]
    $\delta_{L,2^{L-1}}=2\alpha_{L,2^{L-1}}$
}

% ==========================================
% 2. Source
% ==========================================

\node[port] (in_port) at (\srcX, 3.2) {};
\node[src] (source) at (\srcX, 4.4) {Single RF tone\\$s(t)=\sqrt{2}\Re\{\sqrt{P}e^{j2\pi f_ct}\}$};

\draw[->, thick] (source.south) -- node[right, font=\footnotesize] {in 1} (in_port.north);

% Feed the root cell (right input)
\draw[wire] (in_port.south) -- (cRoot_inR);

% ==========================================
% 3. Tree connections (representative)
% ==========================================

% Root -> Level 2
\draw[wire] (cRoot_outL) -- ++(0,-0.2) -| (c21_inR);
\draw[wire] (cRoot_outR) -- ++(0,-0.2) -| (c22_inR);

% Dotted continuation from level 2
\draw[dotwire] (c21_outL) -- ++(0,-1.2);
\draw[dotwire] (c21_outR) -- ++(0,-1.2);
\draw[dotwire] (c22_outL) -- ++(0,-1.2);
\draw[dotwire] (c22_outR) -- ++(0,-1.2);

% Dotted inputs to bottom representative cells
\draw[dotwire] ($(cL1_inR)+(0, 1.2)$) -- (cL1_inR);
\draw[dotwire] ($(cLR_inR)+(0, 1.2)$) -- (cLR_inR);

% Ellipses to indicate omitted subtrees
\node at ($(c21_outR)!0.5!(c22_outL) + (0, -1.8)$) {\Huge $\dots$};
\node at ($(cL1)!0.5!(cLR) + (0, 2.5)$) {\Huge $\dots$};
\node at ($(cL1) + (0, 2.0)$) {\Large $\vdots$};
\node at ($(cLR) + (0, 2.0)$) {\Large $\vdots$};

% ==========================================
% 4. Outputs with per-antenna phase bank (representative)
% ==========================================

\def\outY{-3.6*\levelDist}

\newcommand{\drawOutArrow}[2]{
    % Tap leaf output to a reference height
    \draw[wire] (#1) -- (#1 |- 0, \outY) node[port] (p_#2) {};
    % Per-antenna phase shifter
    \node[ph] (ph_#2) at ($(p_#2)+(0,-0.9)$) {$e^{j\vartheta_{#2}}$};
    \draw[wire] (p_#2) -- (ph_#2.north);
    % Antenna excitation output
    \draw[->, thick] (ph_#2.south) -- ++(0,-0.7) node[below] {$x_{#2}$};
}

\drawOutArrow{cL1_outL}{1}
\drawOutArrow{cL1_outR}{2}
\drawOutArrow{cLR_outL}{N-1}
\drawOutArrow{cLR_outR}{N}

\node at (0, \outY - 1.0) {\Large $\dots$};

% Outer boundary of the programmable passive network
\draw[rounded corners=10pt, thick] (-11.0, 3.2) rectangle (9.0, \outY - 1.7);

\end{tikzpicture}
}
\caption{Proposed fully-analog transmitter architecture as a balanced binary-tree programmable signal distributor (conceptual block diagram, $N=2^L$).
Each block
$\mathrm{Cell}_{\ell,i}$ in the figure denotes one fanout-mode interferometric splitter of
Fig.~\ref{fig:splitter_cell_matched_load}. Thus, the label $\delta_{\ell,i}=2\alpha_{\ell,i}$ is the MZI
differential phase that sets the power split at node $(\ell,i)$. The idle input of every cell is terminated in $Z_0$ to
enforce the zero-incident-wave condition implied by Lemma~\ref{lem:zero_input_property}. After $L$ layers, the
$N$ leaf outputs pass through a per-antenna output phase bank $e^{j\vartheta_n}$ implementing
$\bm P(\bm\vartheta)$, producing the desired antenna excitations $x_1,\dots,x_N$.}
\label{fig:full_layout_tree_block}
\end{figure*}

The construction in \Cref{sec:exact_synthesis} proves that, when $\bm V_{\mathrm{mag}}(\bm\alpha)$ is applied to
the fixed excitation $\bm e_1$, every embedded $2\times2$ splitter is driven with one nonzero input and one
\emph{identically zero} input (Lemma~\ref{lem:zero_input_property}). In the hardware interpretation, we enforce
this operating condition by terminating the corresponding idle input port in a matched load ($Z_0$), so its
incident wave is (ideally) zero and no reflections propagate back into the network.

\begin{remark}[Input-interface variant via fixed concentration]
In some platforms it is convenient to distribute the source using a corporate-feed divider (e.g., for routing)
and then apply a fixed lossless $N$-port that concentrates the resulting uniform excitation into port~1 of the
programmable network. Let
\begin{equation}
\bm u_0 \triangleq \frac{1}{\sqrt{N}}\bm 1,\qquad \|\bm u_0\|_2=1,
\end{equation}
denote the uniform excitation produced by an ideal equal-split divider. Choose a fixed unitary
$\bm U_{\mathrm{in}}\in\mathbb{C}^{N\times N}$ satisfying $\bm U_{\mathrm{in}}\,\bm u_0=\bm e_1$.
Then the programmable network $\bm V(\bm\varphi)\bm U_{\mathrm{in}}$ driven by $\sqrt{P}\bm u_0$
excites $\bm V(\bm\varphi)$ exactly as $\sqrt{P}\bm e_1$, so the programming rule of
\Cref{sec:exact_synthesis} applies without modification. A canonical choice is the normalized $N$-point DFT
(Butler) matrix $\bm F_N$.
\end{remark}

% ------------------------------------------------------------
\subsection{Component count and scalability (idealized)}
\label{subsec:component_count}

The programmable portion of the split-then-phase network consists of:
(i) $N-1$ internal $2\times2$ splitter cells arranged as a balanced binary tree of depth $L=\log_2 N$, and
(ii) an output phase bank of $N$ per-antenna tunable phase shifters.

\smallskip
\noindent\textbf{Per-node cell realization and tunable controls:}
In the MZI-based splitter of Fig.~\ref{fig:splitter_cell_matched_load}, each internal node cell uses:
(i) two fixed $3$\,dB couplers and (ii) one tunable differential phase setting $\delta_{\ell,i}=2\alpha_{\ell,i}$
that controls the split ratio. No additional per-node phase trim is required inside the tree.

\smallskip
\noindent\textbf{Whole-network totals and depth:}
The magnitude tree contains $N-1$ internal cells and therefore uses $2(N-1)$ fixed $3$\,dB couplers. Each
root-to-leaf path traverses exactly $L=\log_2 N$ splitter cells, followed by one output phase shifter. The
tunable degrees of freedom are:
\begin{itemize}
\item $(N-1)$ split controls $\{\delta_{\ell,i}\}$ (equivalently $\{\alpha_{\ell,i}\}$), one per internal node, and
\item $N$ output phases $\{\vartheta_n\}$, one per antenna,
\end{itemize}
for a total of $2N-1$ real tunable controls, matching the minimal parameter count for universal state preparation
with a single driven input mode.

% ============================================================
% Section 5: Non-ideal hardware model and comparison framework
% ============================================================
\section{Non-ideal considerations and normalization framework}
\label{sec:nonideal_framework}

This section defines (i) a minimal passive-contractive abstraction for the realized programmable network and
(ii) a consistent delivered-power normalization boundary that is useful when the proposed fully-analog
architecture is implemented and calibrated in hardware. We emphasize that the present paper does not attempt an
exhaustive impairment-level characterization (e.g., finite phase quantization, coupler/matching errors,
tuning-dependent loss, LUT calibration accuracy, or PA loading effects). Such implementation-specific issues
depend strongly on the chosen phase-control technology and physical layout and are left to future experimental
work.

The development in \Cref{sec:system_model,sec:exact_synthesis,sec:hardware} assumes lossless, perfectly matched
components and exact programmability of (i) the magnitude-tree split settings and (ii) the per-antenna output
phase bank. Practical implementations deviate from this ideal due to insertion loss, amplitude imbalance,
imperfect couplers, static phase offsets, finite phase resolution, and tuning-dependent loss. For the purposes
of specifying calibration targets and comparing efficiency against fully-digital multi-antenna transmitters, it is
sufficient to adopt a \emph{passive contractive} model that separates the synthesized \emph{direction} from the
delivered \emph{scale} and that uses a consistent delivered-power boundary.

\subsection{Passive contractive model and direction error metric}
\label{subsec:nonideal_model_metric}

Let $\bm\eta\in\mathbb{R}^{M}$ denote the vector of hardware control settings (voltages or digital codes) applied
to the tunable elements that realize the split-then-phase architecture: $(N-1)$ split controls in the magnitude
tree and $N$ output phase controls in the phase bank (so $M=2N-1$). In practice, the realized mapping at the
antenna feeds is
\begin{equation}
\hat{\bm x}
\;=\;
\hat{\bm V}(\bm\eta)\,\hat{\bm x}_{\mathrm{in}}
\;=\;
\sqrt{\Pin}\,\hat{\bm V}(\bm\eta)\,\bm e_1,
\label{eq:nonideal_mapping}
\end{equation}
where $\Pin$ is the injected tone power at the driven input port and $\hat{\bm V}(\bm\eta)$ denotes the realized
(generally non-unitary) transfer of the passive interferometric network.

Because the network is passive and (nominally) matched at its external ports, it is non-expansive in the
power-wave sense:
\begin{equation}
\|\hat{\bm x}\|_2^2
=
\Pin\,\big\|\hat{\bm V}(\bm\eta)\bm e_1\big\|_2^2
\;\le\;
\Pin.
\label{eq:passive_contraction}
\end{equation}
It is convenient to separate the normalized transmit-vector \emph{shape} from the delivered \emph{scale}. Define
\begin{equation}
g(\bm\eta)\triangleq \big\|\hat{\bm V}(\bm\eta)\bm e_1\big\|_2 \in (0,1],
\label{eq:g_def}
\end{equation}
and
\begin{equation}
\hat{\bm c}(\bm\eta)
\triangleq
\frac{\hat{\bm V}(\bm\eta)\bm e_1}{\big\|\hat{\bm V}(\bm\eta)\bm e_1\big\|_2},
\qquad
\big\|\hat{\bm c}(\bm\eta)\big\|_2=1.
\label{eq:chat_def}
\end{equation}
Then the realized transmit vector decomposes as
\begin{equation}
\hat{\bm x}
\;=\;
\sqrt{\Pin}\,g(\bm\eta)\,\hat{\bm c}(\bm\eta).
\label{eq:shape_scale_decomp}
\end{equation}

Given a desired transmit vector $\bm x$ with $\|\bm x\|_2^2=P$ and the corresponding normalized direction
$\bm c=\bm x/\sqrt{P}$, a phase-invariant direction mismatch metric is
\begin{equation}
\varepsilon_{\mathrm{dir}}(\bm\eta)
\triangleq 1-\big|\bm c^H\hat{\bm c}(\bm\eta)\big|\in[0,1],
\label{eq:direction_error}
\end{equation}
which equals zero if and only if $\hat{\bm c}(\bm\eta)$ matches $\bm c$ up to a common phase rotation.

In practice, $\varepsilon_{\mathrm{dir}}(\bm\eta)$ is minimized by calibration of the entire device once it is
fabricated. Calibration is performed offline (factory) and occasionally during maintenance, and produces a
look-up table (LUT) that maps \emph{desired} ideal programming parameters
\begin{equation}
\bm\varphi \;\triangleq\; (\bm\alpha,\bm\vartheta)
\end{equation}
(i.e., the magnitude-tree split angles $\bm\alpha=\{\alpha_{\ell,i}\}$ and the output phase-bank settings
$\bm\vartheta=[\vartheta_1,\dots,\vartheta_N]^T$ computed in \Cref{sec:exact_synthesis})
to \emph{hardware} control settings $\bm\eta$. In a future hardware implementation,
$\varepsilon_{\mathrm{dir}}(\bm\eta)$ would be minimized via calibration and LUT-based programming of the
split-control and phase-bank elements.

\subsection{Delivered-power normalization and DC accounting}
\label{subsec:comparison_framework}

To compare transmitter \emph{front ends} at a consistent RF boundary (and to avoid conflating RF hardware with
baseband/compute overhead), we normalize both architectures to the same delivered (conducted) antenna-port RF
power
\begin{equation}
P_{\mathrm{ant,tot}}
\;\triangleq\;
\sum_{n=1}^N P_{n,\mathrm{ant}}
\;=\;
\|\hat{\bm x}\|_2^2.
\label{eq:common_delivered_power_def}
\end{equation}
Here $P_{\mathrm{ant,tot}}$ is the total RF power delivered at the antenna feed ports (conducted power).

For the analog transmitter, \eqref{eq:shape_scale_decomp} implies
\begin{equation}
P_{\mathrm{ant,tot}}
=
\Pin\, g(\bm\eta)^2
\quad\Rightarrow\quad
\Pin = \frac{P_{\mathrm{ant,tot}}}{g(\bm\eta)^2}.
\label{eq:pin_power_choice_nonideal}
\end{equation}
We report the setting-dependent network insertion loss as
\[
L_{\mathrm{net}}(\bm\eta)\triangleq -10\log_{10}\!\big(g(\bm\eta)^2\big)\ \text{dB}.
\]

For DC accounting on the analog side, we use
\begin{equation}
P_{\mathrm{DC,tot}}^{(\mathrm{ana})}
\;\approx\;
\frac{\Pin}{\eta_{\mathrm{PA}}}
\;+\;
P_{\mathrm{DC,ctrl}},
\label{eq:analog_total_dc_normalized}
\end{equation}
where $\eta_{\mathrm{PA}}$ is the effective efficiency of the single PA at output power $\Pin$, and
$P_{\mathrm{DC,ctrl}}$ is the DC overhead of the tunable split-control elements and the output phase-bank
elements (including their bias/control circuitry). In the split-then-phase realization, the number of tunable
controls is $M=2N-1$, so a useful first-order scaling is
\begin{equation}
P_{\mathrm{DC,ctrl}}
\;\approx\;
(2N-1)\,p_{\phi}
\;+\;
p_{\mathrm{ctrl}},
\label{eq:dc_power_ctrl_vd_joint}
\end{equation}
with technology-dependent per-control DC draw $p_{\phi}$ and any non-scaling controller/interface overhead
$p_{\mathrm{ctrl}}$.

\begin{remark}
In the ideal lossless case, the injected tone power equals the delivered antenna-port power, so $\Pin=P$; in the
non-ideal case we distinguish $\Pin$ (injected) from $P_{\mathrm{ant,tot}}$ (delivered).
\end{remark}

For power benchmarking at the same conducted antenna-port boundary $P_{\mathrm{ant,tot}}$ in
\eqref{eq:common_delivered_power_def}, we use a compute-excluded fully-digital RF-front-end DC-power model
anchored to  COTS component data. To keep \Cref{sec:nonideal_framework} focused on the passive-network
normalization and delivered-power boundary, the fully-digital baseline model statement and the sub-6\,GHz
coefficient derivation used in \Cref{sec:numerical_results} are provided in Appendix~\ref{app:fd_sub6_model}.

% ============================================================
% Section 6: Numerical results and practical implications
% ============================================================
\section{Numerical results and practical implications}
\label{sec:numerical_results}

This section consolidates numerical results and the corresponding practical implications for timing feasibility,
passive insertion loss, and RF-front-end DC-power comparison. The loss numbers for the programmable
\emph{balanced binary-tree} network use the per-cell abstraction and stress-case computation summarized in
Appendix~\ref{app:stress_case_loss}. Unless stated otherwise, comparisons are normalized to equal delivered
antenna-port power $P_{\mathrm{ant,tot}}$ as defined in \eqref{eq:common_delivered_power_def}.

All numerical values in this section are intentionally \emph{commercially anchored}: we compare against
\emph{commercial off-the-shelf (COTS)} components and published front-end budgeting examples,
rather than theoretical designs. In particular, the fully-digital baseline uses the \emph{compute-excluded}
RF-front-end DC model in \eqref{eq:digital_dc_decomp}--\eqref{eq:digital_dc_affine}, parameterized from COTS
data, so that the analog and digital power figures are compared at a consistent front-end boundary and do not
depend on platform-specific baseband/compute power.

% ------------------------------------------------------------
\subsection{Symbol-timing budget and OFDM single-subcarrier compatibility}
\label{subsec:symbol_timing_budget}

When fully general symbol-wise synthesis is required, the analog network must be reprogrammed at the symbol
cadence. Let
\begin{equation}
T_{\mathrm{sw}}
\;\triangleq\;
T_{\mathrm{load}} \;+\; T_{\mathrm{tune}} \;+\; T_{\mathrm{settle}}
\label{eq:tsw_def}
\end{equation}
denote the total \emph{reconfiguration-and-settling} time per symbol, where $T_{\mathrm{load}}$ is the
control-interface update time, $T_{\mathrm{tune}}$ is the tuning time of the tunable split/phase elements, and
$T_{\mathrm{settle}}$ captures residual settling. A minimal feasibility condition for symbol-wise operation is
\begin{equation}
\Ts \;\ge\; T_{\mathrm{sw}}.
\label{eq:symbol_timing_budget}
\end{equation}
Microsecond-class tuning supports MHz-class update ceilings in a direct time-domain sense, provided
$T_{\mathrm{load}}$ is kept comparable to (or below) $T_{\mathrm{tune}}$ via interface parallelism.

For an OFDM waveform, the complex coefficient on any \emph{single subcarrier} is constant over one OFDM symbol.
Let $T_{\mathrm{OFDM}}$ denote the OFDM symbol duration (including cyclic prefix if present), i.e.,
\begin{equation}
T_{\mathrm{OFDM}} \;=\; T_{\mathrm{u}} \;+\; T_{\mathrm{CP}},
\label{eq:tofdm_def}
\end{equation}
where $T_{\mathrm{u}}$ is the useful (orthogonality) interval and $T_{\mathrm{CP}}$ is the cyclic-prefix
duration. Typical OFDM numerologies span multiple \emph{time-domain classes} for $T_{\mathrm{u}}$ (and therefore
$T_{\mathrm{OFDM}}$), ranging from a few microseconds to tens of microseconds. For example, using
$T_{\mathrm{u}}=1/\Delta f$ (where $\Delta f$ is the subcarrier spacing), one obtains:
\begin{itemize}
\item \emph{Long-symbol OFDM:} $T_{\mathrm{u}}\approx 66.7\,\mu$s (e.g., $\Delta f=15$\,kHz), so
      $T_{\mathrm{OFDM}}$ is on the order of $\sim 70\,\mu$s including a modest cyclic prefix.
\item \emph{Medium-symbol OFDM:} $T_{\mathrm{u}}\approx 12.8\,\mu$s (e.g., $\Delta f=78.125$\,kHz), so
      $T_{\mathrm{OFDM}}$ is on the order of $\sim 13.6$--$16\,\mu$s depending on cyclic-prefix choice.
\item \emph{Short-symbol OFDM:} $T_{\mathrm{u}}\approx 3.2\,\mu$s (e.g., $\Delta f=312.5$\,kHz), so
      $T_{\mathrm{OFDM}}$ is on the order of $\sim 4\,\mu$s including a typical cyclic prefix.
\end{itemize}

To support OFDM on a \emph{single subcarrier} using symbol-wise reconfiguration, one can set the analog symbol
interval to the OFDM symbol duration, i.e., $\Ts=T_{\mathrm{OFDM}}$, and the timing feasibility condition
\eqref{eq:symbol_timing_budget} becomes
\begin{equation}
T_{\mathrm{sw}} \;\le\; T_{\mathrm{OFDM}}.
\label{eq:tofdm_feasibility}
\end{equation}
In this interpretation, the analog waveform within each OFDM symbol interval can be decomposed into:
(i) a \emph{reconfiguration transient} of duration $T_{\mathrm{sw}}$ during which control updates and settling
occur, followed by (ii) a \emph{steady-state} interval during which the network settings are constant and the
output on each antenna is a pure sinusoid (for that subcarrier). The steady-state duration is
\begin{equation}
T_{\mathrm{ss}}
\;\triangleq\;
T_{\mathrm{OFDM}} - T_{\mathrm{sw}}.
\label{eq:tss_def}
\end{equation}
Hence, when $T_{\mathrm{sw}}$ is \emph{much smaller} than the relevant $T_{\mathrm{OFDM}}$ class, the
fully-analog transmitter can support symbol intervals at least as short as those of the considered OFDM
single-subcarrier waveform, with most of each OFDM symbol interval spent in steady state. Moreover, if
$T_{\mathrm{sw}}$ can be confined within (or below) a cyclic-prefix budget $T_{\mathrm{CP}}$, then the useful
interval $T_{\mathrm{u}}$ can remain effectively undisturbed in an OFDM receiver that discards the cyclic prefix.

In practical implementations, each antenna branch should include appropriate band-pass filtering to suppress any
out-of-band spectral components generated during the control transient in the reconfiguration interval.

The ``Reconf.'' column of Table~\ref{tab:joint_rate_power_summary} lists approximate values of
$T_{\mathrm{sw}}$ for the representative phase-control technologies. The table shows that RF-MEMS
reconfiguration supports long-symbol OFDM (e.g., $\Delta f=15$\,kHz), leaving an approximately
$60\,\mu$s steady-state interval in a $\sim 70\,\mu$s symbol interval. The other phase-control options listed in
Table~\ref{tab:joint_rate_power_summary} support even the shortest-duration OFDM symbols in a direct timing
sense.

The phase-control options in Table~\ref{tab:joint_rate_power_summary} should be interpreted as representative
implementation points rather than as a prescription of a single switch technology. A GaN-switch-based
discrete-delay implementation is attractive when high RF power handling, linearity, and fast switching are
prioritized. Silicon/UltraCMOS FET switch implementations provide a different trade-off, often emphasizing
integration, availability, and control simplicity, while RF-MEMS and digital phase-shifter modules represent
other points in the insertion-loss, speed, and DC-power design space. The proposed architecture is therefore
independent of the specific phase-control technology; the selected COTS entries instantiate representative
loss/speed/power trade-offs.

A discussion of a direct wideband extension of the single-subcarrier OFDM into multiple subcarriers is provided
in Appendix~\ref{appendix:wideband_ofdm}.

\begin{table*}[t]
\centering
\caption{Representative phase-control operating points and the resulting \emph{stress-case} network insertion
loss $L_{\mathrm{net}}=-10\log_{10}g(\bm\eta)^2$ for the split-then-phase balanced binary-tree network
($N=2^L$). Stress case uses the lossless constant-modulus schedule $|c_n|^2=1/N$ (equivalently
$\alpha_{\ell,i}=\pi/4$ for all internal nodes), hybrid excess loss $L_{\mathrm{hyb}}=0.12$\,dB~\cite{ttm_x3c21p2_ds},
and the per-cell loss abstraction summarized in Appendix~\ref{app:stress_case_loss}. The reported
$L_{\mathrm{net}}$ includes the magnitude tree plus the per-antenna output phase bank; for simplicity, the same
per-phase-shifter insertion loss $L_\phi$ is assumed for all tunable phase elements. Values are rounded to
0.1\,dB.}
\label{tab:joint_rate_power_summary}
\footnotesize
\renewcommand{\arraystretch}{1.05}
\setlength{\tabcolsep}{3pt}

\begin{tabularx}{\linewidth}{Y|c|c|c|c|*{4}{c}}
\hline
\textbf{Phase-control option} &
\textbf{Res.} &
\textbf{Reconf.} &
$L_\phi$ (dB) &
$p_{\phi}$ (mW) &
\multicolumn{4}{c}{$L_{\mathrm{net}}$ (dB)}\\
\cline{6-9}
 &  &  &  &  & $N{=}2$ & $N{=}4$ & $N{=}8$ & $N{=}16$\\
\hline
RF-MEMS delay/phase element (product-brief beam-steering shifter example)~\cite{menlo_phase_shifter_brief} &
discrete &
$\sim 10\,\mu$s &
0.2 &
$\sim 0.3$\,mW &
0.5 & 0.9 & 1.2 & 1.6\\
\hline
GaN switch--based discrete delay/phase (two-switch series proxy using TS7225FK SPDT)~\cite{tagore_ts7225fk_ds} &
discrete &
$\sim 0.7\,\mu$s &
0.8 &
$\sim 0.9$\,mW &
1.4 & 2.0 & 2.7 & 3.3\\
\hline
High-power UltraCMOS SPDT switch (two-switch series proxy using PE42823, TX path)~\cite{psemi_pe42823_ds} &
discrete &
$\sim 2\,\mu$s &
1.1 &
$\sim 0.8$\,mW  &
1.9 & 2.6 & 3.4 & 4.1\\
\hline
Low-loss digital phase shifter module (DPS family example)~\cite{nardamiteq_dps05030509} &
6-bit &
$0.2$--$0.5\,\mu$s &
1.4 &
$\sim 250$\,mW &
2.3 & 3.2 & 4.1 & 4.9\\
\hline
\end{tabularx}
\end{table*}

% ------------------------------------------------------------
\subsection{Stress-case insertion loss of the programmable binary tree}
\label{subsec:stress_case_insertion_loss}

Because the synthesis network is passive, the delivered power is $P_{\mathrm{ant,tot}}=\Pin\,g(\bm\eta)^2$,
and therefore higher insertion loss directly inflates the required drive power $\Pin$ via
\eqref{eq:pin_power_choice_nonideal}. Table~\ref{tab:joint_rate_power_summary} reports a conservative stress-case
network insertion loss $L_{\mathrm{net}}=-10\log_{10}g(\bm\eta)^2$, computed using the \emph{lossless
constant-modulus} schedule $|c_n|^2=1/N$ and the per-cell loss model of Appendix~\ref{app:stress_case_loss}.
For a balanced binary tree with equal-size subtrees, the lossless constant-modulus schedule corresponds to
$\alpha_{\ell,i}=\pi/4$ at every internal node.

In a balanced tree with $N=2^L$ leaves, each root-to-leaf path traverses exactly $L=\log_2 N$ programmable
splitting cells and one output phase shifter. As a result, stress-case insertion-loss accumulation along any one
output path scales as $O(\log_2 N)$.

\begin{table}[t]

\scriptsize 
\centering

\caption{Compute-excluded fully-digital multi-antenna transmitter RF-front-end power model coefficients for the sub-6\,GHz
(6\,GHz-class)  COTS regime used in Section~\ref{sec:numerical_results}. The model is
$P_{\mathrm{DC,TX}}^{(\mathrm{dig})}\approx P_{\mathrm{DC,sh}}+\alpha N+\beta P_{\mathrm{ant,tot}}$ with
$P_{\mathrm{DC,sh}}=0$ for transparency. Coefficients and validity range are derived in
Appendix~\ref{app:fd_sub6_model}.}
\label{tab:fd_coeff_summary}

\renewcommand{\arraystretch}{1.2}  
\setlength{\tabcolsep}{2pt}  
\begin{tabularx}{\linewidth}{p{1.2cm}|Y|c|c|Y}
\hline
\textbf{Band} & 
\textbf{Regime} \newline \textbf{(Anchor)} & 
$\boldsymbol{\alpha}$ \newline (W/ch.) & 
$\boldsymbol{\beta}$ \newline (W/W) & 
\textbf{Notes} \\
\hline
sub-6\,GHz \newline (6\,GHz class) &
Wi-Fi FEM (QPF4658 fit) &
2.67 &
3.19 &
Valid for $p_{\mathrm{ant}} \in$ \newline  $[0.08, 0.25]$\,W \\
\hline
\end{tabularx}
\end{table}

% ------------------------------------------------------------
\subsection{Compute-excluded fully-digital baseline (COTS, sub-6\,GHz)}
\label{subsec:fully_digital_baseline_models}

We benchmark the fully-analog front end against a fully-digital multi-antenna transmitter using the compute-excluded
RF-front-end model in \eqref{eq:digital_dc_decomp}--\eqref{eq:digital_dc_affine}. This model deliberately excludes
baseband/compute (precoding, DPD, MAC/PHY, host processing) and isolates the \emph{front-end} DC power required
to deliver a target conducted antenna-port power $P_{\mathrm{ant,tot}}$.

The model parameters are anchored to representative commercial off-the-shelf (COTS) component
data for the sub-6\,GHz (6\,GHz-class) operating point considered in this paper. For conciseness, we summarize
the resulting $(\alpha,\beta)$ coefficients (and their validity range) in \Cref{tab:fd_coeff_summary}. A
self-contained derivation of these coefficients from the underlying datasheet anchors is provided in
Appendix~\ref{app:fd_sub6_model}.

\begin{table*}[t]
\centering
\caption{Compute-excluded equal-$p_{\mathrm{ant}}$ comparison at sub-6\,GHz.
Assumptions: $p_{\mathrm{ant}}=0.2$\,W ($\approx 23$\,dBm) so $P_{\mathrm{ant,tot}}=Np_{\mathrm{ant}}$;
digital: $(\alpha,\beta)=(2.67,\,3.19)$ (Wi-Fi FEM regime; see \Cref{tab:fd_coeff_summary}).
Analog: $\eta_{\mathrm{PA}}=0.5$, $P_{\mathrm{DC,ctrl}}\approx(2N-1)p_\phi$, and $L_{\mathrm{net}}(N)$ from
Table~\ref{tab:joint_rate_power_summary} (stress-case loss model; includes output phase bank under the stated
$L_\phi$ assumption). All values are at the RF-front-end boundary and correspond to COTS anchors.}
\label{tab:equal_power_comparison_sub6}
\scriptsize
\renewcommand{\arraystretch}{1.05}
\setlength{\tabcolsep}{2pt}
\begin{tabularx}{\linewidth}{c|c|c|>{\centering\arraybackslash}X|>{\centering\arraybackslash}X|>{\centering\arraybackslash}X|>{\centering\arraybackslash}X}
\hline
$N$ &
\shortstack{$P_{\mathrm{ant,tot}}$\\(W)} &
\shortstack{Digital\\DC (W)} &
\shortstack{Analog DC\\RF-MEMS (W)} &
\shortstack{Analog DC\\GaN switch (W)} &
\shortstack{Analog DC\\UltraCMOS switch (W)} &
\shortstack{Analog DC\\DPS module (W)} \\
\hline
2  & 0.40 & 6.62  & 0.90 & 1.11 & 1.24 & 2.11 \\
4  & 0.80 & 13.23 & 1.97 & 2.54 & 2.92 & 5.09 \\
8  & 1.60 & 26.46 & 4.22 & 5.97 & 7.01 & 11.98 \\
16 & 3.20 & 52.93 & 9.26 & 13.71 & 16.48 & 27.53 \\
\hline
\end{tabularx}
\end{table*}

\subsection{Component-Role to Representative COTS Anchor Map}
The COTS anchors used in the comparison map directly to the circuit roles in the proposed architecture and in
the fully-digital reference baseline. The fixed $3$\,dB hybrid coupler used in each MZI/fanout splitter is
anchored by the TTM/Xinger--Anaren hybrid data sheet~\cite{ttm_x3c21p2_ds}, whose excess loss
$L_{\mathrm{hyb}}$ is used in the splitter-cell loss model of Appendix~\ref{app:stress_case_loss}. The tunable
split-control and output-phase elements are instantiated by the representative phase-control options in
Table~\ref{tab:joint_rate_power_summary}: the Menlo Micro RF-MEMS delay/phase product
brief~\cite{menlo_phase_shifter_brief}, the TagoreTech TS7225FK GaN SPDT switch used as a two-switch
discrete-delay proxy~\cite{tagore_ts7225fk_ds}, the pSemi PE42823 UltraCMOS SPDT switch used as a
corresponding FET-switch proxy~\cite{psemi_pe42823_ds}, and the Narda-MITEQ digital phase-shifter
module~\cite{nardamiteq_dps05030509}. These anchors provide the phase-control insertion loss $L_{\phi}$,
per-control DC draw $p_{\phi}$, and reconfiguration-time entries used in
Table~\ref{tab:joint_rate_power_summary}. The single-source analog PA term in
\eqref{eq:analog_total_dc_normalized} is anchored by the MACOM SSPA efficiency value
$\eta_{\mathrm{PA}}$~\cite{macom_engad00074_ds}. For the compute-excluded fully-digital reference, the
per-chain RF/transceiver overhead $P_{\mathrm{DC,chain}}$ is anchored by AD9375/ADRV9009 TX-mode
power data~\cite{adi_ad9375_ds,adi_adrv9009_ds}, and the per-antenna PA/FEM affine model is anchored by the Qorvo
QPF4658 Wi-Fi FEM operating points in Appendix~\ref{app:fd_sub6_model}~\cite{qorvo_qpf4658_product}.

% ------------------------------------------------------------
\subsection{Analog--digital comparison at equal $p_{\mathrm{ant}}$ and equal delivered power}
\label{subsec:equal_power_comparison}

We \emph{fix the per-antenna conducted power} $p_{\mathrm{ant}}$ and therefore set
\begin{equation}
P_{\mathrm{ant,tot}} \;=\; N\,p_{\mathrm{ant}}.
\label{eq:optionA_total_power}
\end{equation}
This enforces an equal per-antenna conducted-power operating point across the fully-analog and fully-digital
front ends, and it avoids device-specific aggregation assumptions.

For the analog transmitter, we compute the required injected tone power as
$\Pin=P_{\mathrm{ant,tot}}\cdot 10^{L_{\mathrm{net}}/10}$ using $L_{\mathrm{net}}$ from
Table~\ref{tab:joint_rate_power_summary}, and then evaluate the analog RF-front-end DC power via
\eqref{eq:analog_total_dc_normalized} with\footnote{Because the proposed fully-analog transmitter generates a
single narrowband sinusoid, its PA can be operated close to saturation without the linearity back-off typically
required for high-PAPR modulated waveforms. Commercial sub-6\,GHz GaN SSPAs report about 50\% effective PA
efficiency at saturation in the 4--6\,GHz band; accordingly, for the sub-6\,GHz numerical results we use an
effective PA efficiency of $\eta_{\mathrm{PA}}=0.5$ \cite{macom_engad00074_ds}.}
$\eta_{\mathrm{PA}}=0.5$ and $P_{\mathrm{DC,ctrl}}\approx (2N-1)p_\phi$ (using the $p_\phi$ values of
Table~\ref{tab:joint_rate_power_summary}).

For the fully-digital baseline, we use the compute-excluded COTS model
$P_{\mathrm{DC,TX}}^{(\mathrm{dig})}\approx \alpha N+\beta P_{\mathrm{ant,tot}}$ with coefficients summarized in
\Cref{tab:fd_coeff_summary}. The resulting equal-$p_{\mathrm{ant}}$ comparisons for the sub-6\,GHz (6\,GHz-class)
operating point are reported in \Cref{tab:equal_power_comparison_sub6}.

Table~\ref{tab:equal_power_comparison_sub6} shows that, at the common conducted-power operating point
$p_{\mathrm{ant}}=0.2$\,W ($\approx 23$\,dBm) per antenna, the proposed fully-analog front end requires much less
\emph{compute-excluded} RF-front-end DC power than the fully-digital baseline for all tested array sizes
$N\in\{2,4,8,16\}$ and for all representative phase-control options.

For example, at $N=16$ the compute-excluded fully-digital baseline consumes $52.9$\,W, whereas the fully-analog
architecture consumes $9.26$\,W (RF-MEMS), $13.71$\,W (GaN switch), $16.48$\,W (UltraCMOS switch), and
$27.53$\,W (DPS module). This corresponds to an approximate $48$--$83\%$ reduction in RF-front-end DC power
(equivalently a $1.9\times$--$5.7\times$ improvement) at equal delivered antenna-port power under the stress-case
loss abstraction.

This advantage is driven primarily by eliminating the $N$ RF chains of a fully-digital multi-antenna transmitter (per-chain
transceiver/RFIC overhead that scales linearly with $N$), replacing them with a single PA driving a passive
distribution network; the remaining analog overhead is dominated by the single-PA drive inflation due to passive
insertion loss and by the DC draw of the $(2N-1)$ tunable controls.

From a system viewpoint, these front-end savings enable a clear \emph{bandwidth--energy trade-off}.
Although the proposed transmitter is narrowband (single tone or single OFDM subcarrier in the base architecture),
it preserves the \emph{full $N$-dimensional spatial signaling space} of a fully-digital array by supporting
arbitrary $\bm x[m]$. Hence, standard MIMO gains---beamforming gain and, when the channel supports it, spatial
multiplexing and MU-MIMO---can be exploited without reintroducing the per-antenna RF-chain DC overhead of
fully-digital multi-antenna transmitters. In energy-constrained deployments, this ability to deliver
fully-digital-equivalent MIMO functionality at much lower RF-front-end DC power can help mitigate the throughput
penalty associated with narrowband operation.

We report results up to $N\le 16$ to avoid extrapolating the compute-excluded fully-digital front-end model beyond
the regime where its COTS-anchored coefficients are most representative. For substantially larger arrays,
additional implementation-dependent overheads may introduce costs not captured by the simple affine scaling
$P_{\mathrm{DC,TX}}^{(\mathrm{dig})}\approx \alpha N+\beta P_{\mathrm{ant,tot}}$, potentially leading to
superlinear growth with $N$ in specific architectures. For completeness, Fig.~\ref{fig:compare} shows the
power consumption predicted by all models for $N\leq 4096$.

\begin{figure*}[t]
\centering
\includegraphics[width=\linewidth]{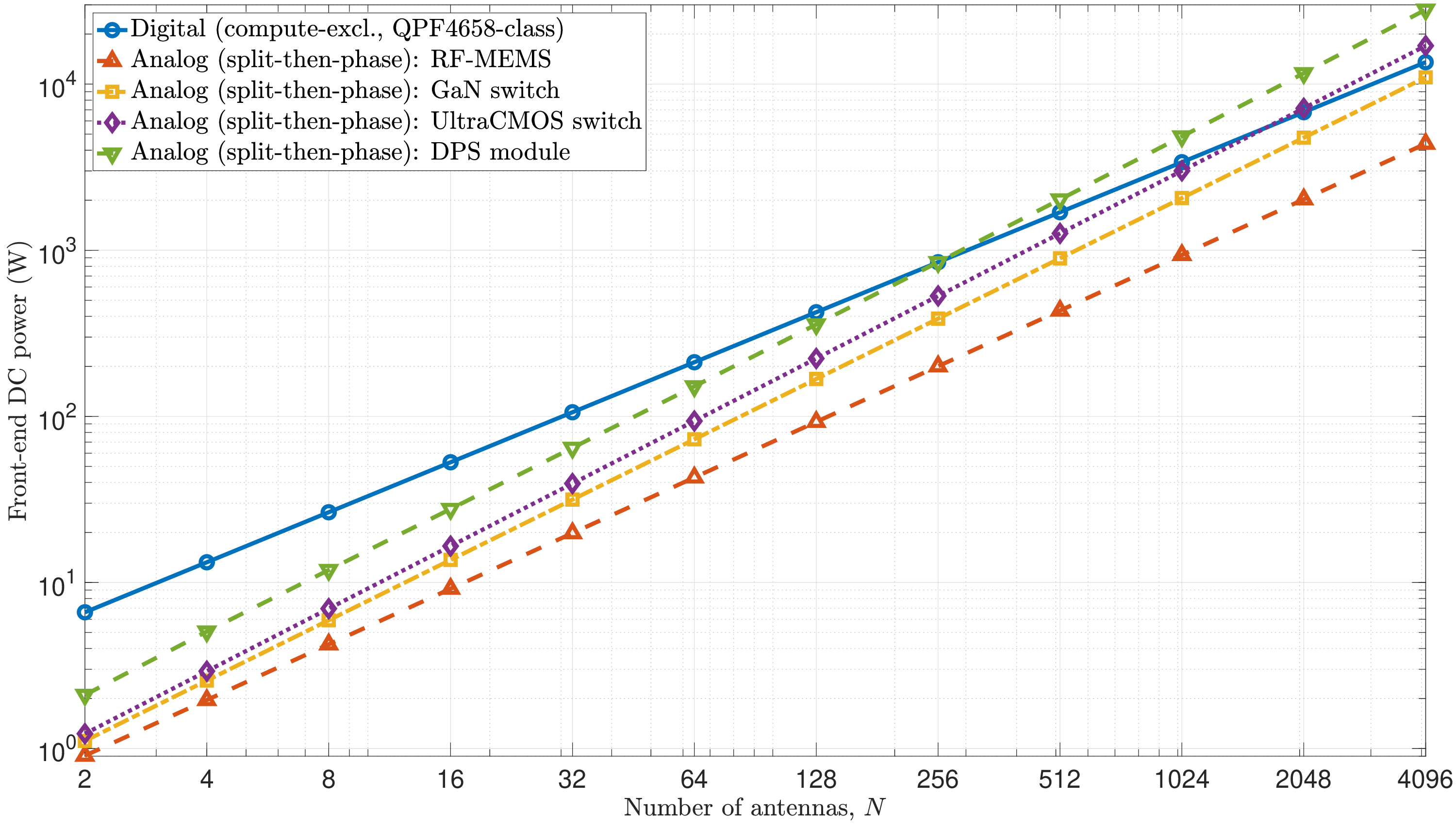}
\caption{Total RF-front-end DC power versus number of antennas $N$ at sub-6\,GHz
(fixed per-antenna conducted power $p_{\mathrm{ant}}=0.2$\,W, so $P_{\mathrm{ant,tot}}=Np_{\mathrm{ant}}$).
Curves compare the compute-excluded fully-digital RF-front-end model against the fully-analog transmitter
parameterized by the stress-case $L_{\mathrm{net}}(N)$ entries in Table~\ref{tab:joint_rate_power_summary}.}
\label{fig:compare}
\end{figure*}

% ------------------------------------------------------------
\subsection{Implications and practical takeaways}
\label{subsec:implications}

Equations \eqref{eq:tofdm_feasibility}--\eqref{eq:tss_def} show that the relevant feasibility question is whether
the reconfiguration transient $T_{\mathrm{sw}}$ is small relative to the intended OFDM symbol duration class
$T_{\mathrm{OFDM}}$. Microsecond-class tuning therefore supports many OFDM operating points directly in the time
domain: it is comfortably below long-symbol OFDM classes (tens of microseconds) and can remain compatible even
with short-symbol OFDM classes (few microseconds) when $T_{\mathrm{sw}}$ is sufficiently small. The steady-state
interval $T_{\mathrm{ss}}=T_{\mathrm{OFDM}}-T_{\mathrm{sw}}$ provides a direct measure of how much of each OFDM
symbol can be spent producing a clean sinusoid under fixed analog settings.

For $N\le 16$, Table~\ref{tab:joint_rate_power_summary} indicates that microsecond-class reconfiguration and
single-digit-dB stress-case insertion loss are simultaneously feasible when the effective per-control insertion
loss is kept near the $\lesssim 1$--$1.5$\,dB regime. The balanced-tree topology ensures that the number of
programmable splitting stages per antenna grows only as $\log_2 N$ (with one additional output phase element).

In \Cref{tab:equal_power_comparison_sub6}, both the analog and digital power figures are evaluated at a
\emph{consistent RF-front-end boundary} and intentionally \emph{exclude} baseband/compute power. The digital
baseline accounts for the per-chain RF/IF hardware overhead plus the PA/FEM contribution needed to deliver
$p_{\mathrm{ant}}$ (compute excluded), while the analog estimate accounts for the single PA drive inflation
caused by passive insertion loss plus the DC draw of the programmable phase-control network. All coefficients and
operating points are anchored to existing COTS components and published commercial front-end budgeting
examples, rather than hypothetical designs.

If one instead seeks \emph{total radio} power (including baseband/compute, platform DC/DC, and system functions),
those terms must be added on top of the front-end figures and will depend strongly on the specific device and
implementation.

The analog transmitter requires symbol-wise programming of $2N-1$ tunable controls. However, the \emph{computation}
to obtain the control settings is closed form and scales as $O(N)$ (cf.~\Cref{sec:exact_synthesis}).
Accordingly, the dominant implementation constraints are typically (i) the tunable-element hardware (tuning
speed, resolution, insertion loss) and (ii) the control-interface loading/settling budget of
\eqref{eq:symbol_timing_budget}, rather than arithmetic complexity.

This numerical comparison is intentionally scoped to sub-6\,GHz operation, where both the phase-control-network
parameters and the compute-excluded fully-digital front-end coefficients can be anchored to publicly available
COTS data at a consistent RF boundary. Extending the same methodology to mmWave is of clear interest, but
doing so credibly requires a similarly consistent set of COTS-anchored insertion-loss, control-power, and
per-chain front-end power coefficients in the intended band and power regime; at present these inputs are less
uniformly available across component classes. A dedicated mmWave comparison using a unified set of validated
anchors and models is therefore left for future work.

% ============================================================
% Section 7: Conclusion and Outlook
% ============================================================
\section{Conclusion and outlook}
\label{sec:conclusion}

This paper developed a narrowband fully-analog $N$-antenna transmitter that reproduces the
\emph{symbol-wise transmit-vector capability} of a narrowband fully-digital $N$-antenna transmitter at the
conducted antenna-port boundary. The proposed architecture is driven by a single coherent RF tone and uses a
passive interferometric network to synthesize an arbitrary desired excitation vector
$\bm x[m]\in\mathbb{C}^N$ with prescribed total conducted power $\|\bm x[m]\|_2^2=P$ in symbol interval $m$. In particular, the synthesized transmit vector may vary arbitrarily from one symbol interval to the next
(subject only to $\|\bm x[m]\|_2^2=P$ and the hardware reconfiguration time).

Crucially, despite its narrowband nature, the proposed front end remains \emph{fully-digital-equivalent in the
spatial domain}: by enabling synthesis of an arbitrary $\bm x[m]$ at the antenna ports, it supports the same
narrowband linear precoding used for spatial multiplexing and multi-user MIMO as fully-digital arrays, but with
substantially lower RF-front-end DC power due to replacing $N$ RF chains by a single source/PA and a passive
unitary distributor.

The synthesis task was posed as a unitary state-preparation problem for the unit-norm direction
$\bm c[m]=\bm x[m]/\sqrt{P}$. The key architectural principle is a \emph{split-then-phase} realization:
a balanced binary-tree \emph{magnitude distributor} first allocates the desired magnitudes
$|c_1[m]|,\dots,|c_N[m]|$ using only tunable lossless $2\times2$ splitter cells operated in fanout mode, and a
per-antenna output phase bank then assigns the desired complex phases. This yields an explicit closed-form,
non-iterative programming rule for the split controls and output phases that achieves exact synthesis in the
ideal lossless model (\Cref{sec:exact_synthesis}). The resulting hardware mapping is compact and scalable: the
tree requires $N-1$ tunable split settings and the phase bank requires $N$ tunable phase settings, totaling
$2N-1$ real degrees of freedom, while using only fixed passive couplers and tunable phase shifters (no
programmable attenuators and no distributed RF gain).

We further introduced a passive contractive non-ideal abstraction and a delivered-power normalization boundary
(\Cref{sec:nonideal_framework}) that separates \emph{direction} accuracy from \emph{delivered} scale and enables a
consistent efficiency comparison against conventional fully-digital RF front ends. Using commercially anchored
COTS parameters and a compute-excluded digital baseline model, the numerical results indicate substantial
RF-front-end DC-power savings for the proposed fully-analog approach for $N\le 16$ at equal delivered
antenna-port power (\Cref{sec:numerical_results}).

Several extensions are of interest. A full hardware demonstration requires calibration and LUT-based
programming under non-unitary effects (insertion loss, static phase offsets, finite resolution, and
tuning-dependent loss). Wideband operation beyond a single narrowband tone motivates architectures that either
parallelize across tones/subcarriers or implement intrinsically broadband passive networks, together with
appropriate filtering to control spectral regrowth during reconfiguration transients.

% ============================================================
% Appendices
% ============================================================
\appendices

\section{MZI transfer matrix and realization of the tunable $2\times2$ splitter}
\label{app:mzi_derivation}

This appendix provides derivations used in \Cref{sec:hardware}:
(i) the transfer matrix of an ideal Mach--Zehnder interferometer (MZI) built from two $3$\,dB couplers and a
differential phase shift, and (ii) its use as a tunable lossless $2\times2$ splitter for the magnitude
distribution tree. We also highlight the fanout operating mode relevant to single-port excitation with one input
matched and unexcited.

\subsection{Derivation of the MZI transfer matrix}

Let $\bm H$ denote the ideal $3$\,dB coupler model in \eqref{eq:hybrid_H}. An MZI with internal differential
phase $\delta$ can be written as
\begin{equation}
\bm U_{\mathrm{MZI}}(\delta)
=
\bm H\,
\mathrm{diag}\!\left(e^{j\delta/2},\,e^{-j\delta/2}\right)
\bm H.
\end{equation}
Substituting \eqref{eq:hybrid_H} and multiplying yields
\begin{equation}
\bm U_{\mathrm{MZI}}(\delta)=
\begin{bmatrix}
\cos(\delta/2) & j\sin(\delta/2)\\
j\sin(\delta/2) & \cos(\delta/2)
\end{bmatrix},
\label{eq:app_mzi_matrix}
\end{equation}
which is unitary for all $\delta$.

\subsection{Equivalence to a tunable fanout splitter}

Define the tunable $2\times2$ splitter
\begin{equation}
\bm U(\alpha)\triangleq
\begin{bmatrix}
\cos\alpha & j\sin\alpha\\
j\sin\alpha & \cos\alpha
\end{bmatrix},
\qquad \alpha\in[0,\pi/2].
\label{eq:app_Ualpha}
\end{equation}
Setting $\delta=2\alpha$ in \eqref{eq:app_mzi_matrix} gives $\bm U_{\mathrm{MZI}}(2\alpha)=\bm U(\alpha)$.
Therefore, an MZI realizes a lossless programmable split ratio via the single differential control $\delta=2\alpha$,
while introducing a \emph{known, fixed} branch phase convention through the factor $j$.

In particular, for any $a\in\mathbb{C}$,
\begin{equation}
\bm U(\alpha)\begin{bmatrix}a\\0\end{bmatrix}
=
\begin{bmatrix}
a\cos\alpha\\
j\,a\sin\alpha
\end{bmatrix},
\label{eq:app_split_fanout}
\end{equation}
so $\alpha$ controls the magnitude split while the relative phase between the two outputs is a fixed $+\pi/2$
offset on the ``$\sin$'' branch under the convention of \eqref{eq:app_Ualpha}.

\subsection{Fanout-mode operation under a matched idle input}

In the single-port excitation regime used in this paper, each internal splitter cell of the balanced tree is
driven in fanout mode with one input carrying the parent subtree wave and the other input terminated in a matched
load (idle input; zero incident wave). Under this operating condition, the magnitude tree allocates the desired
leaf magnitudes by appropriately choosing the split angles $\{\alpha_{\ell,i}\}$, while all target phases are
assigned by a per-antenna output phase bank. Any deterministic fixed phase offsets introduced by layout
(reference-plane conventions, unequal line lengths, or coupler phase conventions) can be absorbed into the output
phase bank via calibration.

\section{Unitarity of embedded splitter cells and disjoint-layer products}
\label{App-B}

This appendix records two facts used by the balanced binary-tree magnitude distributor in \Cref{sec:exact_synthesis}:
(i) any $2\times2$ splitter $\bm U(\alpha)$ embedded into an $N\times N$ identity is unitary, and
(ii) a product of such embeddings acting on disjoint index pairs is unitary (hence each binary-tree layer matrix
is unitary).

\subsection{General embedded splitter $\bm U_{m,n}(\alpha)$}

Let $1\le m<n\le N$. Define $\bm U_{m,n}(\alpha)\in\mathbb{C}^{N\times N}$ as the identity on all coordinates
except the two-dimensional subspace spanned by $\{\bm e_m,\bm e_n\}$, where it acts as $\bm U(\alpha)$ in
\eqref{eq:app_Ualpha}. Entrywise,
\begin{equation}
\big[\bm U_{m,n}(\alpha)\big]_{u,v} =
\begin{cases}
\cos\alpha, & (u,v)\in\{(m,m),(n,n)\},\\[2pt]
j\sin\alpha, & (u,v)\in\{(m,n),(n,m)\},\\[2pt]
1, & u=v\ \text{and}\ u\notin\{m,n\},\\[2pt]
0, & \text{otherwise}.
\end{cases}
\label{eq:Umn_entries}
\end{equation}

\begin{lemma}
\label{lem:Umn_unitary}
For any $1\le m<n\le N$ and any $\alpha\in[0,\pi/2]$, the embedded matrix $\bm U_{m,n}(\alpha)$ is unitary.
\end{lemma}
\begin{IEEEproof}
All columns $v\notin\{m,n\}$ equal standard basis vectors $\bm e_v$ and are mutually orthonormal.
The only nontrivial columns are
\[
\big[\bm U_{m,n}\big]_{\cdot,m}=\cos\alpha\,\bm e_m + j\sin\alpha\,\bm e_n,
\]
\[
\big[\bm U_{m,n}\big]_{\cdot,n}=j\sin\alpha\,\bm e_m + \cos\alpha\,\bm e_n.
\]
Their norms satisfy $\cos^2\alpha+\sin^2\alpha=1$, and their inner product is
\begin{align*}
&\big(\big[\bm U_{m,n}\big]_{\cdot,m}\big)^H\big(\big[\bm U_{m,n}\big]_{\cdot,n}\big)
=
\cos\alpha\,(j\sin\alpha)^* + (j\sin\alpha)\,\cos\alpha^*\nonumber\\
&=
-j\cos\alpha\sin\alpha + j\cos\alpha\sin\alpha
=0.
\end{align*}
Hence the columns form an orthonormal set, so $\bm U_{m,n}^H\bm U_{m,n}=\bm I_N$.
\end{IEEEproof}

\subsection{Products over disjoint pairs (layer matrices)}

\begin{lemma}
\label{lem:disjoint_layer_unitary_splitter}
Let $\{(m_k,n_k)\}_{k=1}^K$ be index pairs such that no index appears in more than one pair (i.e., the pairs are
disjoint). Then the embedded splitters $\bm U_{m_k,n_k}(\alpha_k)$ commute, and the product
\[
\bm U \triangleq \prod_{k=1}^{K} \bm U_{m_k,n_k}(\alpha_k)
\]
is unitary.
\end{lemma}
\begin{IEEEproof}
If two embedded matrices act on disjoint coordinate sets, they act as identity on each other's active
two-dimensional subspaces; hence they commute. Each factor is unitary by Lemma~\ref{lem:Umn_unitary}, and a
product of unitary matrices is unitary. Therefore, $\bm U$ is unitary.
\end{IEEEproof}

Lemma~\ref{lem:disjoint_layer_unitary_splitter} implies that each balanced-tree layer matrix (formed as a
product of disjoint embedded $2\times2$ splitters) is unitary. Since the per-antenna phase bank is diagonal with
unit-modulus entries, it is also unitary. Therefore, the overall ideal mapping (magnitude tree followed by the
output phase bank) is unitary in the lossless matched model.

\section{Stress-case insertion-loss computation for the balanced binary tree}
\label{app:stress_case_loss}

This appendix summarizes the per-cell loss abstraction used to compute the stress-case insertion-loss values in
Table~\ref{tab:joint_rate_power_summary} for the \emph{balanced binary-tree} network with an output phase bank.
The intent is to keep the main body focused on implications while retaining a transparent engineering model for
reproducing the reported numbers.

\subsection{Per-cell and output-bank loss factors}

Let $L_{\mathrm{hyb}}$ denote the excess loss (dB) of one nominal $3$\,dB hybrid coupler and let $L_\phi$ denote
the insertion loss (dB) of one tunable phase-control element. Define the corresponding power transmission
factors
\[
\rho_{\mathrm{2hyb}}\triangleq 10^{-2L_{\mathrm{hyb}}/10},\qquad
\rho_\phi\triangleq 10^{-L_\phi/10}.
\]
Following the same cell-level abstraction as in the main text, each magnitude-tree splitter cell (two hybrids
plus one internal tunable phase element) is approximated by a \emph{common} power transmission factor $\rho_c$
multiplying both outputs:
\begin{equation}
\rho_c \;\approx\; \rho_{\mathrm{2hyb}}\cdot \frac{1+\rho_\phi}{2}.
\label{eq:cell_loss_factor_common}
\end{equation}
The per-antenna output phase bank contributes one additional tunable phase element per output path; denote its
power transmission factor by
\[
\rho_{\mathrm{out}}\triangleq 10^{-L_{\mathrm{out}}/10}.
\]
In Table~\ref{tab:joint_rate_power_summary} we use the simplifying assumption $L_{\mathrm{out}}=L_\phi$ (same
technology class for all tunable phase elements), so $\rho_{\mathrm{out}}=\rho_\phi$.

\subsection{Tree recursion and delivered-power fraction}

Assume $N=2^L$ and index the internal splitter cells by levels $\ell\in\{1,\dots,L\}$, where $\ell=1$ is the
root. Let $p_{\ell,i}$ denote the power entering cell $(\ell,i)$, with unit input power
\begin{equation}
p_{1,1}=1.
\end{equation}
Under the abstraction \eqref{eq:cell_loss_factor_common}, cell $(\ell,i)$ produces the child powers
\begin{align}
p_{\ell+1,2i-1}
&= \rho_c\,p_{\ell,i}\,\cos^2\alpha_{\ell,i},
\nonumber\\
p_{\ell+1,2i}
&= \rho_c\,p_{\ell,i}\,\sin^2\alpha_{\ell,i},
\qquad \ell=1,\dots,L.
\label{eq:bt_power_recursion}
\end{align}
The leaf (antenna-port) powers correspond to the auxiliary level $\ell=L+1$, i.e., $p_{L+1,n}$ for
$n=1,\dots,N$.

The delivered power fraction of the magnitude tree alone is $\sum_{n=1}^{N} p_{L+1,n}$. Including the output
phase bank, the delivered-power fraction becomes
\begin{equation}
g(\bm\eta)^2 \;\approx\; \rho_{\mathrm{out}}\sum_{n=1}^{N} p_{L+1,n},
\label{eq:bt_g2_sum}
\end{equation}
and the corresponding network insertion loss is $L_{\mathrm{net}}=-10\log_{10}(g(\bm\eta)^2)$.

\subsection{Stress-case schedule (lossless constant modulus)}

For the stress case used in Table~\ref{tab:joint_rate_power_summary}, we use the \emph{lossless constant-modulus}
schedule $|c_n|^2=1/N$. For a balanced binary tree with equal-size left and right subtrees at every internal
node, the corresponding ideal split is equal at every node:
\[
\cos^2\alpha_{\ell,i}=\sin^2\alpha_{\ell,i}=\frac{1}{2}
\qquad\Rightarrow\qquad
\alpha_{\ell,i}=\frac{\pi}{4}.
\]
Substituting $\alpha_{\ell,i}=\pi/4$ into \eqref{eq:bt_power_recursion}, the \emph{total} power leaving each node
is multiplied by $\rho_c$ per level. Since the tree has depth $L=\log_2 N$, the magnitude tree delivers the
fraction $\rho_c^{\,L}$. Including the output phase bank gives the closed form
\begin{equation}
g(\bm\eta)^2 \;\approx\; \rho_{\mathrm{out}}\cdot \rho_c^{\,L},
\qquad L=\log_2 N.
\label{eq:bt_g2_closed_form}
\end{equation}
Evaluating \eqref{eq:bt_g2_closed_form} with $L_{\mathrm{hyb}}=0.12$\,dB and the $L_\phi$ values listed in
Table~\ref{tab:joint_rate_power_summary} (and $L_{\mathrm{out}}=L_\phi$) reproduces the reported stress-case
$L_{\mathrm{net}}$ entries (up to the stated rounding).

\section{Compute-excluded fully-digital baseline model and sub-6\,GHz coefficient derivation}
\label{app:fd_sub6_model}

This appendix states the compute-excluded fully-digital RF-front-end DC-power model used for benchmarking and
derives the sub-6\,GHz (6\,GHz-class) coefficients $(\alpha,\beta)$ summarized in \Cref{tab:fd_coeff_summary}.
The model intentionally excludes baseband/compute (precoding, DPD, PHY/MAC, host SoC) and accounts only for the
RF-front-end power required to deliver a target conducted antenna-port power $P_{\mathrm{ant,tot}}$.

\subsection{Compute-excluded fully-digital RF-front-end model (generic form)}
Let $p_{\mathrm{ant}}\triangleq P_{\mathrm{ant,tot}}/N$ denote the per-antenna conducted power. We model the
compute-excluded fully-digital RF-front-end DC power as
\begin{align}
P_{\mathrm{DC,TX}}^{(\mathrm{dig})}(N,P_{\mathrm{ant,tot}})
& \; \approx\;
P_{\mathrm{DC,sh}}
\;+\;
N\,P_{\mathrm{DC,chain}}\nonumber\\
& \;+\;
N\,P_{\mathrm{DC,PA}}\!\left(p_{\mathrm{ant}}\right),
\label{eq:digital_dc_decomp}
\end{align}
where $P_{\mathrm{DC,chain}}$ is the per-chain transceiver/RFIC overhead (DACs, mixers/IQ, LO distribution,
drivers/bias) and $P_{\mathrm{DC,PA}}(\cdot)$ is the PA/FEM DC power needed to deliver $p_{\mathrm{ant}}$ at the
antenna port. The shared term $P_{\mathrm{DC,sh}}$ captures any non-scaling front-end overhead (set to $0$ in
\Cref{sec:numerical_results} for transparency).

Over a relevant operating range, $P_{\mathrm{DC,PA}}(p_{\mathrm{ant}})$ is well-approximated by an affine model
$P_{\mathrm{DC,PA}}(p_{\mathrm{ant}})\approx a + b\,p_{\mathrm{ant}}$, yielding
\begin{equation}
P_{\mathrm{DC,TX}}^{(\mathrm{dig})}(N,P_{\mathrm{ant,tot}})
\;\approx\;
P_{\mathrm{DC,sh}}
\;+\;
\alpha\,N
\;+\;
\beta\,P_{\mathrm{ant,tot}},
\label{eq:digital_dc_affine}
\end{equation}
with $\alpha\triangleq P_{\mathrm{DC,chain}}+a$ and $\beta\triangleq b$.

\subsection{Per-chain transceiver overhead $P_{\mathrm{DC,chain}}$}

To anchor the per-chain RF/IF overhead (DACs, mixers/IQ modulator, LO distribution, and driver/bias), we use
TX-mode power-dissipation specifications of integrated dual-transmitter RF transceivers.

For the AD9375, the datasheet reports a total TX-mode power dissipation of $3.70$\,W with two Tx channels
enabled at $0$\,dB RF attenuation, and $3.11$\,W at $15$\,dB RF attenuation (see Table~2 in
\cite{adi_ad9375_ds}). Similarly, for the ADRV9009 the datasheet reports $3.68$\,W and $3.11$\,W, respectively,
under analogous TX-mode conditions (see Table~2 in \cite{adi_adrv9009_ds}).

These anchors correspond to a per-transmit-chain overhead of roughly $1.56$--$1.85$\,W/chain across devices and
attenuation settings. In the main text we adopt the representative value
\begin{equation}
P_{\mathrm{DC,chain}}(f_c)\;\approx\;P_{\mathrm{chain}}^{(\mathrm{sub6})}\;\triangleq\;1.8~\mathrm{W/chain},
\label{eq:sub6_Pchain_choice}
\end{equation}
which is consistent with the TX-mode dual-Tx dissipation figures above at sub-6\,GHz frequencies.

\subsection{PA/FEM affine model $P_{\mathrm{DC,PA}}(p_{\mathrm{ant}})$ from a 6\,GHz Wi-Fi FEM}

As the representative sub-6\,GHz PA/FEM class, we use the Qorvo QPF4658 6\,GHz Wi-Fi FEM. The vendor product
table provides typical transmit supply currents at $V_{\mathrm{CC}}=5$\,V for several modulated output-power
operating points \cite{qorvo_qpf4658_product}. The three points used here are summarized in
\Cref{tab:qpf4658_points}. The conducted antenna-port power is
$p_{\mathrm{ant}}=10^{(P_{\mathrm{out,dBm}}-30)/10}$\,W and the corresponding DC power is
$P_{\mathrm{DC}}=V_{\mathrm{CC}}I_{\mathrm{CC}}$.

\begin{table}[t]
\centering
\caption{Representative QPF4658 transmit operating points used to anchor the per-antenna PA/FEM DC model
(typical values from the vendor product table \cite{qorvo_qpf4658_product}).}
\label{tab:qpf4658_points}
\footnotesize
\renewcommand{\arraystretch}{1.05}
\setlength{\tabcolsep}{4pt}
\begin{tabular}{c|c|c|c|c}
\hline
\textbf{Mode} & $P_{\mathrm{out}}$ (dBm) & $p_{\mathrm{ant}}$ (W) & $I_{\mathrm{CC}}$ (mA) & $P_{\mathrm{DC}}$ (W)\\
\hline
11ax & 19 & 0.0794 & 225 & 1.125\\
11ac & 23 & 0.1995 & 300 & 1.500\\
11n  & 24 & 0.2512 & 335 & 1.675\\
\hline
\end{tabular}
\end{table}

Over this operating range, we fit an affine model
\begin{equation}
P_{\mathrm{DC,PA}}(p_{\mathrm{ant}})\;\approx\; a + b\,p_{\mathrm{ant}},
\qquad p_{\mathrm{ant}}\in[0.0794,\,0.2512]~\mathrm{W}.
\label{eq:qpf4658_affine_model}
\end{equation}
With $K$ data points $\{(p_k,P_k)\}_{k=1}^{K}$, the closed-form least-squares fit is
\begin{align}
b &\;=\; \frac{K\sum_k p_kP_k - \big(\sum_k p_k\big)\big(\sum_k P_k\big)}
              {K\sum_k p_k^2-\big(\sum_k p_k\big)^2},
\nonumber\\
a &\;=\; \frac{1}{K}\sum_k P_k \;-\; b\,\frac{1}{K}\sum_k p_k.
\label{eq:affine_fit_closed_form}
\end{align}
Evaluating \eqref{eq:affine_fit_closed_form} for \Cref{tab:qpf4658_points} yields
\begin{equation}
a \approx 0.870~\mathrm{W},
\qquad
b \approx 3.188~\mathrm{W/W}.
\label{eq:qpf4658_fit_ab}
\end{equation}

\subsection{Resulting $(\alpha,\beta)$ coefficients}

Substituting \eqref{eq:sub6_Pchain_choice} and \eqref{eq:qpf4658_fit_ab} into the affine fully-digital
front-end model \eqref{eq:digital_dc_affine} with $P_{\mathrm{DC,sh}}=0$ gives
\begin{align}
\alpha &\;\triangleq\; P_{\mathrm{DC,chain}}+a \;\approx\; 1.8+0.870 \;\approx\; 2.67~\mathrm{W/chain},
\nonumber\\
\beta &\;\triangleq\; b \;\approx\; 3.19~\mathrm{W/W},
\label{eq:sub6_alpha_beta_result}
\end{align}
valid for $p_{\mathrm{ant}}=P_{\mathrm{ant,tot}}/N \in [0.0794,\,0.2512]$\,W.

\section{Extension to multi-subcarrier OFDM via parallelization}
\label{appendix:wideband_ofdm}

A direct wideband extension of the single-subcarrier discussion is to realize each active OFDM subcarrier with
its own dedicated copy of the proposed narrowband passive distributor and its own dedicated input tone.
Specifically, let $\mathcal{K}$ denote the set of active subcarriers with tone frequencies $f_c+f_k$. For each
$k\in\mathcal{K}$, drive an independent programmable passive network
$\hat{\bm V}^{(k)}(\bm\eta^{(k)})$ with an input tone of power $\Pin^{(k)}$ at $f_c+f_k$, producing the
per-subcarrier antenna excitation vector
$\hat{\bm x}^{(k)}[m]=\sqrt{\Pin^{(k)}}\,\hat{\bm V}^{(k)}(\bm\eta^{(k)})\bm e_1$.
The radiated waveform at antenna $n$ is then the multitone sum
$x_n(t)=\sqrt{2}\Re\{\sum_{k\in\mathcal{K}} \hat{x}^{(k)}_n[m]\,e^{j2\pi(f_c+f_k)t}\}$ over the OFDM symbol interval.
Because OFDM subcarriers are orthogonal, the average conducted antenna-port power over the symbol adds across
subcarriers,
$P_{\mathrm{ant,tot}}=\sum_{k\in\mathcal{K}}\|\hat{\bm x}^{(k)}[m]\|_2^2$.

With the passive contraction factor on subcarrier $k$ given by
$g_k\triangleq\|\hat{\bm V}^{(k)}(\bm\eta^{(k)})\bm e_1\|_2$, each tone must satisfy
$\|\hat{\bm x}^{(k)}[m]\|_2^2=\Pin^{(k)}g_k^2$ (cf.~\eqref{eq:pin_power_choice_nonideal}), hence
$\Pin^{(k)}=\|\hat{\bm x}^{(k)}[m]\|_2^2/g_k^2$ and the total injected RF power is
$P_{\mathrm{in,tot}}=\sum_k \Pin^{(k)}$.
If the replicated per-subcarrier networks have approximately equal insertion loss across the occupied band (the
intended operating assumption of this conceptual construction), i.e., $g_k\approx g$ for all
$k\in\mathcal{K}$ (equivalently $L_{\mathrm{net}}^{(k)}\approx L_{\mathrm{net}}$), then
\begin{align}
P_{\mathrm{in,tot}}
&=\sum_{k\in\mathcal{K}}\frac{\|\hat{\bm x}^{(k)}[m]\|_2^2}{g^2}
=\frac{1}{g^2}\sum_{k\in\mathcal{K}}\|\hat{\bm x}^{(k)}[m]\|_2^2
\nonumber\\
&=\frac{P_{\mathrm{ant,tot}}}{g^2}
= P_{\mathrm{ant,tot}}\cdot 10^{L_{\mathrm{net}}/10},
\end{align}
so the dominant PA term in the analog front-end DC accounting \eqref{eq:analog_total_dc_normalized} depends on
the same total delivered power $P_{\mathrm{ant,tot}}$ and the same insertion-loss inflation factor as in the
single-tone case; in this sense, splitting the waveform into multiple dedicated narrowband tone paths does not
change the aggregate required RF drive (and thus PA DC) at fixed $P_{\mathrm{ant,tot}}$.

This per-subcarrier parallelization is a simple conceptual extension, and for a practical OFDM system that uses
many closely spaced subcarriers, it would require a large number of replicated networks. An interesting future
research direction is the development of a more hardware-efficient wideband architecture that retains the exact
narrowband synthesis property while avoiding per-subcarrier replication.

\bibliographystyle{IEEEtran}
\bibliography{biblio_final}

\end{document}